\documentclass[10pt,twoside,reqno,a4paper]{article}
\usepackage[T1]{fontenc}
\usepackage{a4wide}

\usepackage{amsfonts, amsmath, amssymb,latexsym,nccmath}
\usepackage[activeacute,english]{babel}
\usepackage{epsfig}
\usepackage{graphicx}
\usepackage[font=sf,sc]{caption}
\usepackage{float}
\usepackage{cancel}
\usepackage{color}
\usepackage[utf8]{inputenc}
\usepackage{amsthm}
\usepackage{bm}
\usepackage[all]{xy}
\usepackage{enumerate}
\usepackage{comment}
\usepackage[colorinlistoftodos]{todonotes} 
\usepackage{bbm}

\usepackage[colorlinks=true,linkcolor=blue,urlcolor=blue,citecolor=blue]{hyperref}

\newcommand{\R}{{\mathbb {R}}}

\renewcommand{\leq}{\leqslant}
\renewcommand{\geq}{\geqslant}

\newcommand{\RR}{\mathbb{R}}

\renewcommand{\Re}{\operatorname{Re}}
\renewcommand{\Im}{\operatorname{Im}}

\newtheorem{theorem}{Theorem}
\newtheorem{thm}[theorem]{Theorem}
\theoremstyle{plain}

\newtheorem{definition}[theorem]{Definition}

\newtheorem{lem}[theorem]{Lemma}

\newtheorem{prop}[theorem]{Proposition}
\newtheorem{remark}[theorem]{Remark}
\newtheorem{obs}[theorem]{Observation}

\numberwithin{equation}{section}
\numberwithin{theorem}{section}

\allowdisplaybreaks

\theoremstyle{definition}

\usepackage{enumerate}

\usepackage[
backend=bibtex,
style=numeric,
sorting=nyt,
giveninits=true, 
doi=false,isbn=false,url=false,eprint=true,sortcites=true,
]{biblatex}
\DeclareNameAlias{author}{family-given}
\appto{\bibsetup}{\sloppy}
\usepackage{csquotes}
\RequirePackage{dsfont} 
\newcommand{\EE}{\mathbb{E}} 
\newcommand{\PP}{\mathbb{P}} 
\newcommand{\QQ}{\mathbb{Q}} 
\newtheorem{ass}{Assumption}

\title{Pricing VIX options under the Heston-Hawkes stochastic volatility model}

\author{Oriol Zamora Font\thanks{Department of Mathematics, University of Oslo, P.O. Box 1053 Blindern, N-0316 Oslo, Norway, Email: oriolz@math.uio.no}}
\date{\today}

\begin{document}
\maketitle

\begin{abstract}
We derive a semi-analytical pricing formula for European VIX call options under the Heston-Hawkes stochastic volatility model introduced in \cite{arxiv}. This arbitrage-free model incorporates the volatility clustering feature by adding an independent compound Hawkes process to the Heston volatility. Using the Markov property of the exponential Hawkes an explicit expression of $\text{VIX}^2$ is derived as a linear combination of the variance and the Hawkes intensity. We apply qualitative ODE theory to study the existence of some generalized Riccati ODEs. Thereafter, we compute the joint characteristic function of the variance and the Hawkes intensity exploiting the exponential affine structure of the model. Finally, the pricing formula is obtained by applying standard Fourier techniques.


{\it AMS classification MSC2020:} 60G55, 60H10, 91G15, 91G20. 
\end{abstract}

\section{Introduction}\label{Intro}
Volatility trading has become remarkably important in finance in recent years and plays a crucial role in risk management, portfolio diversification, asset pricing and econometrics. The Chicago Board Options Exchange (CBOE) established the first volatility index, the $\text{VXO}$, in January 1993. According to \cite{whaley}, the $\text{VXO}$ index represents the implied volatility of a hypothetical at-the-money option on the S\&P 100 (OEX) expiring in 30 days. The CBOE introduced in September 2003 a new volatility index, the $\text{VIX}$, in a model-free procedure. The $\text{VIX}$ is computed as a weighted sum of out-of-the money S\&P 500 (SPX) calls and puts expiring in 30 days across all available strikes. A convenient property of the $\text{VIX}$ index is that its square approximates the 30-day conditional risk neutral expectation of the return variance, that is, the 30-day variance swap, see \cite{carrandwu}. Many further volatility indices have been introduced like the $\text{VXN}$, the $\text{VXD}$, the $\text{VCAC}$ and others \cite{intro1}. 

\cite{derivatives1} and \cite{derivatives2} described in 1989 and 1993, respectively, the need of volatility derivatives to allow traders to better manage their portfolio and hedge against volatility risk. In that regard, the CBOE Futures Exchange (CFE) introduced $\text{VIX}$ futures contracts in March 2004 followed by the launch of VIX options in February 2006. Thereafter, many new trading opportunities became available for traders giving them more control on volatility risk. The trading popularity of volatility derivatives has consistently increased through the years. This increase is, to some extent,  because $\text{VIX}$ movements are negatively correlated with movements in SPX returns.

Due to the importance of volatility trading in finance, a considerable amount of research has been conducted to develop and expand stochastic volatility models for pricing VIX derivatives. In particular, including jumps in the stock and/or the volatility is a common practice. For instance, \cite{intro1} shows that a mean reverting logarithmic diffusion with Poisson jumps approximates the behaviour of the VIX index and an analytical pricing formula for European VIX options is derived. \cite{Sepp1} uses a Heston model with Poisson jumps in the volatility to derive a pricing formula for volatility derivatives and \cite{Sepp2} for the Heston model with Poisson jumps both in returns and volatility. Also, \cite{zhulian,zhulian2} obtain an exact solution for the pricing of VIX futures in the Heston model with simultaneous Poisson jumps in the asset price and the volatility.

An important reason to go beyond the Poisson jump diffusion models is that the volatility clustering feature and the high-volatility effect \cite{self1,SELF2} are not included due to the constant intensity of the Poisson process. In that regard, \cite{state1,state2,VIXJUMPS1,intro1} obtain pricing formulas for VIX derivatives under a stochastic volatility model with jumps in the asset and/or volatility with state-dependent intensity. Additionally, \cite{VIXcluster} proposes a Heston model with Hawkes jumps in the stock price and \cite{VIXJUMPS1} a Heston model with Hawkes jumps both in the asset and the volatility, respectively. The main purpose of these models is to include the jump clustering effect in the stock and/or the volatility, that is, the occurrence of a jump increases the probability of future jumps. Other fundamental works about the VIX and volatility derivatives are \cite{OTHER1,OTHER2,OTHER3,OTHER4,OTHER5,OTHER6}.

The objective of our work is to obtain a semi-analytical pricing formula for European VIX call options under the Heston-Hawkes stochastic volatility model introduced in \cite{arxiv}. Essentially, the Heston-Hawkes model is a Heston model with a compound Hawkes process in the volatility. In reference to that, the SVCIJ-H model in \cite{VIXJUMPS1} and their approach present some similarities but also some important differences with our paper that are thoroughly described in Section \ref{sec3}. Fundamentally, the SVCIJ-H model in \cite{VIXJUMPS1} has simultaneous and independent Hawkes jumps in the stock and volatility and the size of the jumps is exponentially distributed. In contrast, the self-exciting jumps are only present in the volatility of the Heston-Hawkes model but their size has a general distribution, being the exponential distribution a particular case. This leads to technical challenges in qualitative ODE theory which are properly tackled in Section \ref{CF}. 

First and foremost, the existence of risk neutral probability measures is crucial to make sense of the pricing problem. To the best of our knowledge, the usual procedure in the aforementioned literature is to present the model directly under $\QQ$ after assuming that the model is already arbitrage-free. Nevertheless, \cite{BS96,Ryd99} present some stochastic volatility models where no equivalent local martingale measure exist. Moreover, in terms of risk management practices, \cite{Stein16,Stein162} show that exposures computed under the risk neutral measure are essentially arbitrary and the passage from $\PP$ to $\QQ$ is required. Therefore, a proper study about the existence of risk neutral probability measures is fundamental for the pricing of VIX derivatives. In that respect, \cite{arxiv} proves that the Heston-Hawkes stochastic volatility model is arbitrage-free and an explicit family of equivalent martingale measures is given. 

Our approach to derive a semi-analytical pricing formula is based on the use of the following equality 
\begin{align}\label{formula}
    \EE\left[\max\{\sqrt{X}-K,0\}\right]=\frac{1}{2\sqrt{\pi}}\int_0^\infty\Re\left[\frac{1-\text{erf}\left(K\sqrt{\phi}\right)}{\phi^{3/2}}\EE\left[e^{\phi X}\right]\right]d\phi_I,
\end{align}
where $\text{erf}$ is the error function and some technical conditions described in Section \ref{sec3} must hold. In order to use the previous formula, an explicit expression of $\text{VIX}^2$ is required. Using the Markov property of the exponential Hawkes, it is possible to derive such an explicit expression as a linear combination of the variance and the Hawkes intensity. Another requirement for using the equality is a closed expression of the joint characteristic function of the variance and the Hawkes intensity. Thanks to the exponential affine structure of the model, the characteristic function can be computed after solving some generalized Riccati ODEs. In addition, to apply \eqref{formula} it is crucial to check that the domain of the characteristic function intersects the complex numbers with strictly positive real part. 

The paper is organized as follows. In Section \ref{sec2} we outline the Heston-Hakwes stochastic volatility model introduced in \cite{arxiv}. In Section \ref{sec3} we present the problem of pricing European $\text{VIX}$ call options and the technical results that are required to derive a semi-analytical pricing formula. In Section \ref{prel} we summarize the arbitrage-free property of the stochastic volatility model proven in \cite{arxiv} and some results on the compensators of the Hawkes process under the risk neutral probability measures given in \cite{arxiv2}. In Section \ref{CF} we compute the joint characteristic function of the variance and the Hawkes intensity using the exponential affine structure of the model. Before, we study the solution of some generalized Riccati ODEs that appear in the expression of the characteristic function. In Section \ref{EXPVIX} we derive an explicit expression of the $\text{VIX}$ index using the Markov property of the exponential Hawkes. Finally, in Section \ref{SEC7} we present the semi-analytical pricing formula for European VIX call options under the Heston-Hawkes stochastic volatility model. In the appendix we give all the technical proofs. 

\section{Heston-Hawkes stochastic volatility model}\label{sec2}
We outline the Heston-Hawkes stochastic volatility model introduced in \cite{arxiv}. Essentially, this model is an extension of the well-known Heston model that incorporates the volatility clustering effect by adding an independent compound Hawkes process to the variance.

Let $T\in \R$, $T>0$ be a fixed time horizon. On a complete probability space $(\Omega,\mathcal{A},\PP)$, we consider a two-dimensional standard Brownian motion $(B,W)=\left\{\left(B_t,W_t\right), t\in[0,T]\right\}$ and its minimally augmented filtration $\mathcal{F}^{(B,W)}=\{\mathcal{F}^{(B,W)}_t, t\in[0,T]\}$. On $(\Omega,\mathcal{A},\PP)$, we also consider a Hawkes process $N=\{N_t, t\in[0,T]\}$ with stochastic intensity given by
\begin{align*}
    \lambda_t=\lambda_0+\alpha\int_0^t e^{-\beta(t-s)}dN_s,
\end{align*}
or, equivalently, 
\begin{align}\label{dynlambda}
    d\lambda_t=-\beta(\lambda_t-\lambda_0)dt+\alpha dN_t,
\end{align}
where $\lambda_0>0$ is the initial intensity, $\beta>0$ is the speed of mean reversion and $\alpha\in(0,\beta)$ is the self-exciting factor. Note that the stability condition $\alpha<\beta$ holds. See \cite[Section 2]{Hawkesfinance} and \cite[Section 3.1.1]{article1} for further details on Hawkes processes. 
Then, we consider a sequence of i.i.d., strictly positive and integrable random variables $\{J_i\}_{i\geq 1}$ and the compound Hawkes process $L=\{L_t, t\in[0,T]\}$ given by
\begin{align*}
    L_t=\sum_{i=1}^{N_t}J_i.
\end{align*}   
We make the following assumption on the distribution of the jumps size. 

\begin{ass}\label{as} 
There exists $\epsilon_J>0$ such that the moment generating function $M_J(t)=\EE[e^{tJ_1}]$ of $J_1$ is well defined in $(-\infty,\epsilon_J)$. Moreover, $(-\infty,\epsilon_J)$ is the maximal domain in the sense that
\begin{align*}
    \lim_{t\rightarrow\epsilon_J^-}M_J(t)=\infty.
\end{align*}
Since $\epsilon_J>0$, all positive moments of $J_1$ are finite. Note that it is a rather mild assumption, and that the exponential distribution belongs to the class of distributions satisfying the condition. 
\end{ass}
We assume that $(B,W), N$ and $\{J_i\}_{i\geq 1}$ are independent of each other. We write $\mathcal{F}^L=\left\{\mathcal{F}^L_t, t\in[0,T]\right\}$ for the minimally augmented filtration generated by $L$ and
\begin{align*}
    \mathcal{F}=\{\mathcal{F}_t=\mathcal{F}_t^{(B,W)}\vee\mathcal{F}_t^L, t\in[0,T]\},
\end{align*}
for the joint filtration. We assume that $\mathcal{A}=\mathcal{F}_T$ and we will work with $\mathcal{F}$. Since $(B,W)$ and $L$ are independent processes, $(B,W)$ is also a two-dimensional $(\mathcal{F},\PP)$-Brownian motion. 

Finally, with all these ingredients, we introduce the Heston-Hawkes stochastic volatility model. We assume that the interest rate is deterministic and constant equal to $r$. The stock price $S=\{S_t, t\in[0,T]\}$ and its variance $v=\{v_t, t\in[0,T]\}$ are given by
\begin{align}\label{2}
    \frac{dS_t}{S_t} & =\mu_tdt+\sqrt{v_t}\left(\sqrt{1-\rho^2} dB_t+\rho dW_t\right), \\
    \label{21} dv_t & =-\kappa\left(v_t-\Bar{v}\right)dt+\sigma\sqrt{v_t}dW_t+\eta dL_t,
\end{align}
where $S_0>0$ is the initial price of the stock, $\mu: [0,T] \rightarrow \R$ is a measurable and bounded function, $\rho\in(-1,1)$ is the correlation factor, $v_0>0$ is the initial value of the variance, $\kappa>0$ is the variance's mean reversion speed, $\Bar{v}>0$ is the long-term variance, $\sigma>0$ is the volatility of the variance and $\eta>0$ is a scaling factor. We assume that the Feller condition $2\kappa\Bar{v}\geq\sigma^2$ is satisfied, see \cite[Proposition 1.2.15]{AlfonsiAurélien2015ADaR}. For more details and further results on the model see \cite{arxiv}.

\section{Pricing VIX options}\label{sec3}
The objective of this paper is to derive a semi-analytical pricing formula for European $\text{VIX}$ call options under the Heston-Hawkes stochastic volatility model. Namely, let $\QQ$ be a risk neutral probability measure, the price $\mathcal{C}^\QQ(t,T,K)$ at time $t\in[0,T]$ of an European $\text{VIX}$ call option with maturity time $T$ and strike $K$ is given by
\begin{align}\label{call}
    \mathcal{C}^\QQ(t,T,K)=e^{-r(T-t)}\EE^{\QQ}\left[\max\{\text{VIX}^\QQ_T-K,0\}|\mathcal{F}_t\right],
\end{align}
where the VIX index can be defined (with mathematical simplification) in the following way
    \begin{align*}
        \left(\text{VIX}_t^\QQ\right)^2=-\frac{2}{\Delta}\EE^{\QQ}\left[\log\left(e^{-r\Delta}\frac{S_{t+\Delta}}{  S_t}\right)\Big|\mathcal{F}_t\right] \cdot 100^2,
    \end{align*}
where $\Delta=\frac{30}{365}$ and $S$ follows the dynamics in \eqref{2}. See \cite{VIX1,state1,VIX3,zhulian} for further details on the definition of the VIX index. Our approach to obtain the pricing formula involves the use of Fourier transform techniques. Precisely, we will employ the following convenient equality
\begin{align}\label{fourier}
    \EE\left[\max\{\sqrt{X}-K,0\}\right]=\frac{1}{2\sqrt{\pi}}\int_0^\infty\Re\left[\frac{\text{erfc}\left(K\sqrt{\phi}\right)}{\phi^{3/2}}\EE\left[e^{\phi X}\right]\right]d\phi_I,
\end{align}
where $\phi=\phi_R+i\phi_I\in\mathbb{C}$ with $\phi_R>0$, $X$ is a positive random variable such that $\EE\left[e^{\phi_RX}\right]<\infty$ and $\text{erfc}$ is the complementary error function defined by $\text{erfc}(z)=1-\text{erf}(z)$. For a reference of this formula see \cite{VIXJUMPS1, for1, for2}. The usefulness of this formula relies on the fact that we can derive an explicit expression of $\text{VIX}^2$ and its characteristic function. Several technical issues arise if the formula given in \eqref{fourier} is going to be applied to derive a pricing formula for \eqref{call}:
\begin{enumerate}
    \item The existence of risk neutral probability measures is required to make sense of \eqref{call}. To the best of our knowledge, the common procedure in the literature regarding pricing of VIX options under stochastic volatility models with jumps is to give the model directly under $\QQ$ after assuming that the model is arbitrage-free. It is discussed in \cite{BS96,Ryd99}  that this property is not obvious in general and a proper study is required. \cite{arxiv} proves that the Heston-Hawkes stochastic volatility is arbitrage-free and an explicit family of equivalent martingale measures is given. This is crucial to make sense of this problem and we summarize the results of \cite{arxiv} in Section \ref{prel}. Moreover, \cite{Stein16,Stein162} show that exposures computed under the risk neutral measure are essentially arbitrary and the passage from $\PP$ to $\QQ$ is crucial for risk management practices. 
    
    \item An explicit expression of $\text{VIX}^2$ is necessary  in order to use \eqref{fourier}. In Section \ref{EXPVIX} we obtain such expression by exploiting the Markov property of the exponential Hawkes \cite[Remark 1.22]{ETH}. 
    
    \item There must exist $\phi=\phi_R+i\phi_I\in\mathbb{C}$ with $\phi_R>0$ such that $\EE\left[e^{\phi_RX}\right]<\infty$ to use the formula in \eqref{fourier}. In our setting, this condition is reduced to the existence of $\phi,\psi\in\mathbb{C}$ with $\Re(\phi)>0$ and $\Re(\psi)>0$ such that $\EE^\QQ\left[e^{\phi v_T+\psi\lambda_T}\right]<\infty$ where $\QQ$ is a risk neutral probability measure. This study is conducted in Section \ref{CF}. 
    
    \item A closed expression of the characteristic function $\EE\left[e^{\phi X}\right]$ is needed to obtain a tractable pricing formula. In Section \ref{CF} we derive an (almost) explicit expression of the conditional characteristic function $\EE^\QQ\left[e^{\phi v_T+\psi\lambda_T}|\mathcal{F}_t\right]$ for $t\in[0,T]$ using the exponential affine structure of the model \cite{duffie,VIXJUMPS1}.
\end{enumerate}

After presenting the model, the problem and the approach, it is pertinent to thoroughly explain the connection of our paper with \cite{VIXJUMPS1}.
\begin{remark}
    It is important to mention the connection of our paper with \cite{VIXJUMPS1}. While the objective is the same and the stochastic volatility models are similar, there is a significant amount of differences that are worth to highlight. 
    
    The model SVCIJ-H presented in \cite{VIXJUMPS1} has simultaneous and independent self-exciting jumps in the stock and the variance and their size is exponentially distributed. In our model, the self-exciting jumps occur only in the variance and their size has a general distribution, being the exponential distribution a particular case. 
    
    The stochastic volatility models presented in \cite{VIXJUMPS1} are given directly under $\QQ$ while we justify the existence of a family of equivalent martingale measures for our model. As commented before, the assumption that the model is arbitrage-free may not be obvious \cite{BS96,Ryd99} and the passage from $\PP$ to $\QQ$ is necessary for risk management procedures like computing exposures \cite{Stein16,Stein162}. 
    
    Since the jumps in our setting can have an arbitrary distribution, the ODEs that appear when computing the characteristic function take a more general form than the ones in \cite[Proposition 2]{VIXJUMPS1}. In addition, in Section \ref{CF} we prove that the domain of the characteristic function intersects the complex numbers with strictly positive real part, which is crucial to use the formula in \eqref{fourier}. However, we miss this study in  \cite[Proposition 2]{VIXJUMPS1} for the SVCIJ-H model.
\end{remark}

\section{Arbitrage-free property of the model and compensators}\label{prel}
We summarize the results proven in \cite{arxiv} regarding the existence and positivity of the variance process and the change of measure. In addition, we outline the results in \cite{arxiv2} about the compensator of the Hawkes process under the risk neutral probability measures.

The existence and positivity results of the variance are given in Proposition \ref{p2} and the change of measure is compacted in Theorem \ref{risk}. The results about the compensator of the Hawkes process under the risk neutral measure are given in Proposition \ref{oldprop}.
\begin{prop}\label{p2}
Equation \eqref{21} has a pathwise unique strong solution and the solution is strictly positive.
\end{prop}
\begin{proof}
    See \cite[Proposition 2.1]{arxiv} and \cite[Proposition 2.2]{arxiv}.
\end{proof}

In the next result we define the constant $c_l$, which depends on the model parameters. Essentially, it is a strictly positive constant such that for $c<c_l$ the following holds
\begin{align*}
    \EE\left[\exp\left(c\int_0^Tv_udu\right)\right]<\infty.
\end{align*}
This constant appears in Theorem \ref{risk}, which is the main result regarding the change of measure. 

\begin{prop}\label{novikov}
For $c\leq\frac{\kappa^2}{2\sigma^2}$, define $D(c):=\sqrt{\kappa^2-2\sigma^2c}$, $\Lambda(c):=\frac{2\eta c\left(e^{D(c)T}-1\right)}{D(c)-\kappa+\left(D(c)+\kappa\right)e^{D(c)T}}$ and
\begin{align*}
    c_l:=\sup\left\{c\leq\frac{\kappa^2}{2\sigma^2}: \Lambda(c)< \epsilon_J \hspace{0.3cm}\text{and}\hspace{0.3cm} M_J\left(\Lambda(c)\right)\leq\frac{\beta}{\alpha}\exp\left(\frac{\alpha}{\beta}-1\right) \right\}.
\end{align*}
Then, $c_l>0$ and for $c<c_l$
\begin{align*}
    \EE\left[\exp\left(c\int_0^Tv_udu\right)\right]<\infty.
\end{align*}
\end{prop}
\begin{proof}
    See \cite[Lemma 3.1]{arxiv} and \cite[Proposition 3.5]{arxiv}.
\end{proof}
\begin{thm}\label{risk}Let $a\in\RR$ and define $\theta_t^{(a)}:=\frac{1}{\sqrt{1-\rho^2}}\left(\frac{\mu_t-r}{\sqrt{v_t}}-a\rho\sqrt{v_t}\right)$, 
\begin{align*}
    Y_t^{(a)} &:=\exp\left(-\int_0^t\theta_u^{(a)}dB_u-\frac{1}{2}\int_0^t(\theta_u^{(a)})^2du\right), \\
    Z_t^{(a)} &:=\exp\left(-a\int_0^t\sqrt{v_u}dW_u-\frac{1}{2}a^2\int_0^tv_udu\right)
\end{align*}
and $X_t^{(a)}:=Y_t^{(a)}Z_t^{(a)}$. Recall the definition of $c_l$ in Proposition \ref{novikov}. \begin{enumerate}
    \item $X^{(a)}$ is a $(\mathcal{F},\PP)$-martingale for $|a|<\sqrt{2c_l}$.
    \item The set 
\begin{align*}
    \mathcal{E}:=\left\{\QQ(a) \hspace{0.2cm}\text{given by}\hspace{0.2cm} \frac{d\QQ(a)}{d\PP}=X_T^{(a)} \hspace{0.2cm}\text{with}\hspace{0.2cm} |a|<\sqrt{2c_l}\right\} 
\end{align*}
is a set of equivalent local martingale measures.
\item Let $\QQ(a)\in\mathcal{E}$, the process $(B^{\QQ(a)},W^{\QQ(a)})$ defined by
    \begin{align*}
    dB_t^{\QQ(a)}&:=dB_t+\theta_t^{(a)}dt, \notag\\
    dW_t^{\QQ(a)}&:=dW_t+a\sqrt{v_t}dt
\end{align*}
is a a two-dimensional standard $(\mathcal{F},\QQ(a))$-Brownian motion.
\item Let $\QQ(a)\in\mathcal{E}$, the dynamics of $S$ and $v$ are given by
\begin{align}
    \frac{dS_t}{S_t} &=rdt+\sqrt{v_t}\left(\sqrt{1-\rho^2} dB_t^{\QQ(a)}+\rho dW_t^{\QQ(a)}\right), \nonumber \\
     \label{volqa}dv_t &= -\kappa^{(a)}(v_t-\Bar{v}^{(a)})dt+\sigma\sqrt{v_t}dW_t^{\QQ(a)}+\eta dL_t.
\end{align}
where $\kappa^{(a)}=\kappa+a\sigma>0$ and $\Bar{v}^{(a)}=\frac{k\Bar{v}}{k+a\sigma}>0$.
\item If $\rho^2<c_l$, $\EE^{\QQ(a)}\left[\exp\left(\frac{\rho^2}{2}\int_0^Tv_udu\right)\right]<\infty$ and the set 
\begin{align*}
\mathcal{E}_m:=\left\{\QQ(a)\in\mathcal{E}: |a|<\min\left\{\frac{\sqrt{2c_l}}{2},\sqrt{c_l-\rho^2}\right\}\right\}
\end{align*}
is a set of equivalent martingale measures.
\end{enumerate}
\end{thm}
\begin{proof}
    See \cite[Theorem 3.6]{arxiv}, \cite[Observation 3.8]{arxiv} and \cite[Theorem 3.9]{arxiv}. Note that $\kappa^{(a)}=\kappa+a\sigma>0$ because $|a|<\sqrt{2c_l}\leq\frac{\kappa}{\sigma}$ since $c_l\leq\frac{\kappa^2}{2\sigma^2}$.
\end{proof}

In order to work with the compensators of $N$ and $L$ under the risk neutral measures we make some assumptions and define a (non-empty) subset of equivalent martingale measures. For further details see \cite{arxiv2}.

\begin{ass}\label{as3}
Define $D:=\sup_{t\in[0,T]}\left(\mu_t-r\right)^2<\infty$ and fix $\varepsilon_1,\varepsilon_2>0$. We assume that
\begin{enumerate}
\item $\rho^2<c_l$. 
    \item $\frac{1-\rho^2}{D\left[(2+\varepsilon_1)^2-(2+\varepsilon_1)\right]}\left(\frac{2\kappa\Bar{v}-\sigma^2}{2\sigma}\right)^2>1$.
    \item $2\kappa\Bar{v}>(1+\varepsilon_2)\sigma^2$.
\end{enumerate}
\end{ass}
\begin{definition}
  Let $q,s>1$, we define a subset of $\mathcal{E}_m$ by
    \begin{align*}
        \mathcal{E}_{m}(q,s):=\left\{\QQ(a)\in\mathcal{E}_m:|a|<\min\Bigg\{\frac{1}{qs}\sqrt{\frac{c_l}{2}},\sqrt{\frac{(1-\rho^2)c_l}{qs\left[2qs(1-\rho^2)+\rho^2s-1\right]}}\Bigg\}\right\}.
    \end{align*}
\end{definition}
From now on, we fix $Q_2$ such that $1<Q_2<\frac{1-\rho^2}{D\left[(2+\varepsilon_1)^2-(2+\varepsilon_1)\right]}\left(\frac{2\kappa\Bar{v}-\sigma^2}{2\sigma}\right)^2$ and $Q_1:=\frac{Q_2}{Q_2-1}>1$. The subset $\mathcal{E}_{m}(Q_1,2+\varepsilon_1)$ of equivalent martingale measures is of special interest and will be used throughout the paper thanks to the properties stated in the next proposition.

\begin{prop}\label{oldprop}
    Let $\QQ(a)\in\mathcal{E}_{m}(Q_1,2+\varepsilon_1)$. Then, 
    \begin{enumerate}
    \item $\EE[(X_t^{(a)})^{2+\varepsilon_1}]<\infty$, $\EE\left[\left(\frac{1}{v_t}\right)^{1+\varepsilon_2}\right]<\infty$ for all $t\in[0,T]$ and $\int_0^T\EE\left[\left(\frac{1}{v_t}\right)^{1+\varepsilon_2}\right]dt<\infty$.
        \item $N-\Lambda^N$ is a $(\mathcal{F},\QQ(a))$-martingale.
    \item $L-\Lambda^L$ is a $(\mathcal{F},\QQ(a))$-martingale.
    \end{enumerate}
\end{prop}
\begin{proof}
    See \cite[Lemma 3.7]{arxiv2} and \cite[Propositon 3.12]{arxiv2}.
\end{proof}

\section{Joint characteristic function}\label{CF}
The objective of this section is the computation of the conditional characteristic function of $(v_T,\lambda_T)$. Furthermore, we check that its domain intersects the complex numbers with strictly positive real part, which is required to use the formula in \eqref{fourier}. The subsequent results are based on qualitative ODE theory and the exponential affine structure of the model \cite{duffie,VIXJUMPS1}. 

The following result implies that $\EE[f(J_1)]=\EE^{\QQ(a)}[f(J_1)]$ where $f$ is a deterministic function such that $\EE[|f(J_1)|]<\infty$. Therefore, whenever there is an expectation under $\QQ(a)$ of a random variable of the type $f(J_1)$, it can be equivalently computed under $\PP$.

\begin{lem} Let $\QQ(a)\in\mathcal{E}_m(Q_1,2+\varepsilon_1)$, then $P_{J_1}=\QQ(a)_{J_1}$, that is, the law of $J_1$ is the same under $\PP$ and under $\QQ(a)$. 
\end{lem}
\begin{proof}
See Lemma \ref{AJQP} in the appendix. 
\end{proof}
\begin{definition}\label{LJdef}
    Define the constant $L_J:=\frac{1}{\eta}M_J^{-1}\left(\frac{\beta}{\alpha}\exp\left(\frac{\alpha}{\beta}-1\right)\right)$. Note that $L_J>0$ because $\frac{\beta}{\alpha}\exp\left(\frac{\alpha}{\beta}-1\right)>1$ using $\alpha<\beta$. 
\end{definition}

In the following result we apply qualitative ODE theory to study the existence of some generalized Riccati ODEs. The solutions of such equations appear in the joint characteristic function of $(v_T,\lambda_T)$. A crucial conclusion of the next lemma is that the ODEs admit as initial conditions complex numbers with strictly positive real part. This is a requirement to use the formula in \eqref{fourier}. 

\begin{lem}\label{ODEs} Let $\QQ(a)\in\mathcal{E}_m(Q_1,2+\varepsilon_1)$ and $\phi,\psi\in\mathbb{C}$ such that
\begin{align*}
    0<\Re(\phi)<\min\Bigg\{\frac{2\kappa^{(a)}}{\sigma^2\left(2e^{\kappa^{(a)}T}-1\right)},L_J\Bigg\} \hspace{1cm}\text{and}\hspace{1cm} \Re(\psi)<\frac{\beta-\alpha}{\alpha\beta}.
\end{align*}
Then,
\begin{enumerate}[(i)]
    \item The ODE 
    \begin{align*}
        \frac{d}{dt}G^{(a)}(t;\phi)-\kappa^{(a)}G^{(a)}(t;\phi)+\frac{1}{2}\sigma^2G^{(a)}(t;\phi)^2=0, \hspace{0.5cm} G^{(a)}(T;\phi)=\phi,
    \end{align*}
    has a unique solution in the interval $[0,T]$ given by
    \begin{align*}
        G^{(a)}(t;\phi)=\frac{2\kappa^{(a)}}{\sigma^2+e^{\kappa^{(a)}(T-t)}\left(\frac{2\kappa^{(a)}}{\phi}-\sigma^2\right)}.
    \end{align*}
    \item $\sup_{t\in[0,T]}\Re(G^{(a)}(t;\phi))=\Re(\phi)$.
    \item $\sup_{t\in[0,T]}\EE\left[e^{\eta \Re(G^{(a)}(t;\phi)) J_1}\right]<\frac{\beta}{\alpha}\exp\left(\frac{\alpha}{\beta}-1\right)$. In particular, $\Big|\EE\left[e^{\eta G^{(a)}(t;\phi) J_1}\right]\Big|<\infty$ for all $t\in[0,T]$. 
    \item The ODE 
    \begin{align*}
         \frac{d}{dt}H^{(a)}(t;\phi,\psi)-\beta H^{(a)}(t;\phi,\psi)+e^{\alpha H^{(a)}(t;\phi,\psi)}\EE\left[e^{\eta G^{(a)}(t;\phi) J_1}\right]-1=0, \hspace{0.5cm} H^{(a)}(T;\phi,\psi)=\psi
    \end{align*}
    has a unique solution in $[0,T]$. 
   \item The ODE
   \begin{align*}
        \frac{d}{dt}F^{(a)}(t;\phi,\psi)+\kappa^{(a)}\Bar{v}^{(a)}G^{(a)}(t;\phi)+\beta\lambda_0H^{(a)}(t;\phi,\psi)=0, \hspace{0.5cm} F^{(a)}(T;\phi,\psi)=0.
   \end{align*}
   has a unique solution in $[0,T]$. 
\end{enumerate}
\end{lem}
\begin{proof}
    See Lemma \ref{AODEs} in the appendix. 
\end{proof}

To conclude this section, we give an (almost) closed expression of the joint characteristic function of $(v_T,\lambda_T)$ in terms of the ODEs studied in the previous result. In addition, a subset of the domain of the characteristic function is given which intersects the complex numbers with strictly positive real part. 
\begin{prop}\label{cf}
  Let $\QQ(a)\in\mathcal{E}_m(Q_1,2+\varepsilon_1)$ and $\phi,\psi\in\mathbb{C}$ such that
\begin{align*}
    0<\Re(\phi)<\min\Bigg\{\frac{2\kappa^{(a)}}{\sigma^2\left(2e^{\kappa^{(a)}T}-1\right)},L_J\Bigg\} \hspace{1cm}\text{and}\hspace{1cm} \Re(\psi)<\frac{\beta-\alpha}{\alpha\beta}.
\end{align*}
Then,
\begin{enumerate}
    \item $\Big|\EE^{\QQ(a)}\left[e^{\phi v_T+\psi \lambda_T}\right]\Big|<\infty$.
    \item The joint conditional characteristic function of $(v_T,\lambda_T)$ is given by
   \begin{align*}
      \EE^{\QQ(a)}\left[e^{\phi v_T+\psi \lambda_T}|\mathcal{F}_t\right]=e^{F^{(a)}(t;\phi,\psi)+G^{(a)}(t;\phi)v_t+H^{(a)}(t;\phi,\psi)\lambda_t}
   \end{align*} 
where $t\in[0,T]$, and $F^{(a)},G^{(a)}$ and $H^{(a)}$ are solutions to the following ODEs 
\begin{align}\label{ODEs50}
    \frac{d}{dt}G^{(a)}(t;\phi)-\kappa^{(a)}G^{(a)}(t;\phi)+\frac{1}{2}\sigma^2G^{(a)}(t;\phi)^2=0,& \hspace{0.5cm}G^{(a)}(T;\phi)=\phi, \notag\\
    \frac{d}{dt}H^{(a)}(t;\phi,\psi)-\beta H^{(a)}(t;\phi,\psi)+e^{\alpha H^{(a)}(t;\phi,\psi)}\EE\left[e^{\eta G^{(a)}(t;\phi) J_1}\right]-1=0,& \hspace{0.5cm}H^{(a)}(T;\phi,\psi)=\psi, \notag\\
    \frac{d}{dt}F^{(a)}(t;\phi,\psi)+\kappa^{(a)}\Bar{v}^{(a)}G^{(a)}(t;\phi)+\beta\lambda_0H^{(a)}(t;\phi,\psi)=0,& \hspace{0.5cm}F^{(a)}(T;\phi,\psi)=0.
\end{align}
\end{enumerate}
\end{prop}

\begin{proof}
    (i) First, we consider the case $\Im(\phi)=\Im(\psi)=0$. For $t\in[0,T]$ define the process $M^{(a)}(t):=g^{(a)}(t,v_t,\lambda_t)$ where $g^{(a)}(t,x,y):=e^{F^{(a)}(t;\phi,\psi)+G^{(a)}(t;\phi)x+H^{(a)}(t;\phi,\psi)y}$ and $F^{(a)},G^{(a)}$ and $H^{(a)}$ are the solutions of the ODEs given in Lemma \ref{ODEs}, namely, 
\begin{align}\label{ODEs1}
    \frac{d}{dt}G^{(a)}(t;\phi)-\kappa^{(a)}G^{(a)}(t;\phi)+\frac{1}{2}\sigma^2G^{(a)}(t;\phi)^2=0,& \hspace{0.5cm}G^{(a)}(T;\phi)=\phi, \notag\\
    \frac{d}{dt}H^{(a)}(t;\phi,\psi)-\beta H^{(a)}(t;\phi,\psi)+e^{\alpha H^{(a)}(t;\phi,\psi)}\EE\left[e^{\eta G^{(a)}(t;\phi) J_1}\right]-1=0,& \hspace{0.5cm}H^{(a)}(T;\phi,\psi)=\psi, \notag\\
    \frac{d}{dt}F^{(a)}(t;\phi,\psi)+\kappa^{(a)}\Bar{v}^{(a)}G^{(a)}(t;\phi)+\beta\lambda_0H^{(a)}(t;\phi,\psi)=0,& \hspace{0.5cm}F^{(a)}(T;\phi,\psi)=0.
\end{align}
which are well defined for all $t\in[0,T]$. To lighten the notation we omit the dependence of $g^{(a)}$ on $\phi$ and $\psi$. Note that $F^{(a)},G^{(a)}, H^{(a)}$ take values in $\RR$ because we have assumed that $\Im(\phi)=\Im(\psi)=0$. By convenience we define $Y_t=(t,v_t,\lambda_t)$ and we apply Itô formula to the process $M^{(a)}$ 
\begin{align}\label{expr1}
    M^{(a)}(t) = \ & M^{(a)}(0) +\int_0^t\partial_tg^{(a)}(Y_s)ds+\int_0^t\partial_xg^{(a)}(Y_{s-})dv_s+\int_0^t\partial_yg^{(a)}(Y_{s-})d\lambda_s \notag\\ 
    &+\frac{1}{2}\int_0^t\partial_{xx}^2g^{(a)}(Y_{s-})d[v]^{\text{c}}_s+\frac{1}{2}\int_0^t\partial_{yy}^2g^{(a)}(Y_{s-})d[\lambda]_s^{\text{c}}+\int_0^t\partial_{xy}^2g^{(a)}(Y_{s-})d[v,\lambda]_s^{\text{c}} \notag\\ 
    &+\sum_{0<s\leq t}\left[g^{(a)}(Y_s)-g^{(a)}(Y_{s-})-\partial_xg^{(a)}(Y_{s-})\Delta v_s-\partial_yg^{(a)}(Y_{s-})\Delta\lambda_s\right] \notag\\ 
    = \ & M^{(a)}(0) + \int_0^t \Big[\partial_tg^{(a)}(Y_s)-\kappa^{(a)}\left(v_s-\Bar{v}^{(a)}\right)\partial_xg^{(a)}(Y_s)-\beta\left(\lambda_s-\lambda_0\right)\partial_yg^{(a)}(Y_s) \notag\\
    &+\frac{\sigma^2}{2}v_s\partial_{xx}^2g^{(a)}(Y_s)\Big]ds+\sigma\int_0^t\sqrt{v_s}\partial_xg^{(a)}(Y_{s-})dW_s^{\QQ(a)} \notag\\
    &+\eta\int_0^t\partial_xg^{(a)}(Y_{s-})dL_s+\alpha\int_0^t\partial_yg^{(a)}(Y_{s-})dN_s \notag\\
    &+\sum_{0<s\leq t}\left[g^{(a)}(Y_s)-g^{(a)}(Y_{s-})-\partial_xg^{(a)}(Y_{s-})\Delta v_s-\partial_yg^{(a)}(Y_{s-})\Delta\lambda_s\right] \notag\\
    = \ & M^{(a)}(0) + \int_0^t \Big[\partial_tg^{(a)}(Y_s)-\kappa^{(a)}\left(v_s-\Bar{v}^{(a)}\right)\partial_xg^{(a)}(Y_s)-\beta\left(\lambda_s-\lambda_0\right)\partial_yg^{(a)}(Y_s) \notag\\
    &+\frac{\sigma^2}{2}v_s\partial_{xx}^2g^{(a)}(Y_s)\Big]ds 
    +\sigma\int_0^t\sqrt{v_s}\partial_xg^{(a)}(Y_{s-})dW_s^{\QQ(a)}+\sum_{0<s\leq t}\left[g^{(a)}(Y_s)-g^{(a)}(Y_{s-})\right].
\end{align}
We have used the dynamics of $v$ given in \eqref{volqa}, the dynamics of $\lambda$ given in \eqref{dynlambda} and that 
\begin{align*}
    [v]_t & =\sigma^2\int_0^tv_sds+\eta^2[L]_t \implies [v]_t^{\text{c}}=\sigma^2\int_0^tv_sds, \\
    [\lambda]_t & = \alpha^2 [N]_t = \alpha^2N_t \implies [\lambda]_t^{\text{c}}=0, \\
    [v,\lambda]_t & = \alpha\eta [L,N]_t = \alpha\eta L_t \implies [v,\lambda]_t^{\text{c}}=0, \\
    \eta\int_0^t\partial_xg^{(a)}(Y_{s-})dL_s &= \eta \sum_{0<s\leq t}\partial_xg^{(a)}(Y_{s-})\Delta L_s =  \sum_{0<s\leq t}\partial_xg^{(a)}(Y_{s-})\Delta v_s, \\
    \alpha\int_0^t\partial_yg^{(a)}(Y_{s-})dN_s &= \alpha\sum_{0<s\leq t}\partial_yg^{(a)}(Y_{s-})\Delta N_s = \sum_{0<s\leq t}\partial_yg^{(a)}(Y_{s-})\Delta \lambda_s.
\end{align*}
Next, we can write
\begin{align*}
    \sum_{0<s\leq t}\left[g^{(a)}(Y_s)-g^{(a)}(Y_{s-})\right] & =\sum_{0<s\leq t}\left[g^{(a)}(s,v_{s-}+\eta\Delta L_s, \lambda_{s-}+\alpha\Delta N_s)-g^{(a)}(s,v_{s-},\lambda_{s-})\right] \\ 
    & = \sum_{0<s\leq t}h^{(a)}(s,\Delta L_s, \Delta N_s),
\end{align*}
where 
\begin{align*}
    h^{(a)}(s,u_1,u_2):=g^{(a)}(s,v_{s-}+\eta u_1, \lambda_{s-}+\alpha u_2)-g^{(a)}(s,v_{s-},\lambda_{s-}).
\end{align*}
We now define $R_s=(L_s,N_s)$ and for $t\in[0,T]$, $B\in\mathcal{B}\left(\RR^2\setminus\{0,0\}\right)$, 
\begin{align*}
    N^{R}(t,A)=\#\{0<s\leq t : \Delta R_s\in A\}.
\end{align*}
We add and subtract the compensator of the counting measure $N^R$ to split the expression in \eqref{expr1} into a $(\mathcal{F},\QQ(a))$-local martingale plus a predictable process of finite variation. By Proposition \ref{oldprop}, the compensators of $N$ and $L$ under $\QQ(a)\in\mathcal{E}_m(Q_1,2+\varepsilon_1)$ are $\Lambda_t^N=\int_0^t\lambda_udu$ and $\Lambda_t^L=\EE[J_1]\int_0^t\lambda_udu$ respectively. Thus, 
\begin{align*}
    \sum_{0<s\leq t}\left[g^{(a)}(Y_s)-g^{(a)}(Y_{s-})\right]= \ & \int_0^t\int_{(0,\infty)^2} h^{(a)}(s,u_1,u_2) N^R(ds,du) \\
     = & \int_0^t\int_{(0,\infty)^2} h^{(a)}(s,u_1,u_2) \left(N^R(ds,du)-\lambda_sP_{J_1}(du_1)\delta_1(du_2)ds\right) \\
     & + \int_0^t\lambda_s\int_{(0,\infty)}h^{(a)}(s,u_1,1)P_{J_1}(du_1)ds.
\end{align*}
Replacing everything in \eqref{expr1} we get
\begin{align}\label{expr2}
    M^{(a)}(t) = \ & M^{(a)}(0) + \int_0^t \Big[\partial_tg^{(a)}(Y_s)-\kappa^{(a)}\left(v_s-\Bar{v}^{(a)}\right)\partial_xg^{(a)}(Y_s)-\beta\left(\lambda_s-\lambda_0\right)\partial_yg^{(a)}(Y_s) \notag\\
    &+\frac{\sigma^2}{2}v_s\partial_{xx}^2g^{(a)}(Y_s)+\lambda_s\int_{(0,\infty)}h^{(a)}(s,u_1,1)P_{J_1}(du_1)\Big]ds+\sigma\int_0^t\sqrt{v_s}\partial_xg^{(a)}(Y_{s-})dW_s^{\QQ(a)} \notag\\ 
    &+\int_0^t\int_{(0,\infty)^2} h^{(a)}(s,u_1,u_2) \left(N^R(ds,du)-\lambda_sP_{J_1}(du_1)\delta_1(du_2)ds\right), 
\end{align}
where the last two terms are $(\mathcal{F},\QQ(a))$-local martingales. Now, using the definition of $g^{(a)}$ and $h^{(a)}$ we get
\begin{align*}
    \partial_tg^{(a)}(Y_s) & =g^{(a)}(Y_s)\left(\frac{d}{dt}F^{(a)}(s;\phi,\psi)+\frac{d}{dt}G^{(a)}(s;\phi)v_s+\frac{d}{dt}H^{(a)}(s;\phi,\psi)\lambda_s\right), \\ 
    \partial_xg^{(a)}(Y_s) & = g^{(a)}(Y_s)G(s), \\
    \partial_{xx}^2g^{(a)}(Y_s) & = g^{(a)}(Y_s)G(s)^2, \\
    \partial_yg^{(a)}(Y_s) & =g^{(a)}(Y_s)H(s), \\
    \int_{(0,\infty)}h^{(a)}(s,u_1,1)P_{J_1}(du_1) & = g^{(a)}(Y_{s-})\int_{(0,\infty)}\left[e^{\eta u_1 G(s)+\alpha H(s)}-1\right]P_{J_1}(du_1) \\
    & = g^{(a)}(Y_{s-})\left(e^{\alpha H(s)}\EE\left[e^{\eta G(s) J_1}\right]-1\right).
\end{align*}
Using that $F^{(a)},G^{(a)}$ and $H^{(a)}$ satisfy the ODEs given in \eqref{ODEs1} we see that the drift in \eqref{expr2} is equal to $0$ because
\begin{gather*}
    \frac{d}{dt}F^{(a)}(t;\phi,\psi)+\frac{d}{dt}G^{(a)}(t;\phi)v_t+\frac{d}{dt}H^{(a)}(t;\phi,\psi)\lambda_t \\-\kappa^{(a)}\left(v_t-\Bar{v}^{(a)}\right)G^{(a)}(t)-\beta\left(\lambda_t-\lambda_0\right)H^{(a)}(t)
    +\frac{\sigma^2}{2}v_tG(t)^2+\lambda_t\left(e^{\alpha H(t)}\EE\left[e^{\eta G(t) J_1}\right]-1\right) =0, 
\end{gather*}
We conclude that $M^{(a)}$ is a $(\mathcal{F},\QQ(a))$-local martingale. Moreover, using the final conditions of the ODEs in \eqref{ODEs1}, we see that $M^{(a)}(T)=e^{\phi v_T + \psi\lambda_T}$. Since $M^{(a)}$ is bounded from below it is a supermartingale and the following holds
\begin{align*}
    \EE^{\QQ(a)}\left[e^{\phi v_T+\psi \lambda_T}\right]=\EE^{\QQ(a)}\left[M^{(a)}(T)\right]\leq M^{(a)}(0)=\exp\left(F^{(a)}(0)+G^{(a)}(0)v_0+H^{(a)}(0)\lambda_0\right)<\infty. 
\end{align*}
For the general case of arbitrary $\Im(\phi)$ and $\Im(\psi)$ note that 
\begin{align*}
    \Big|\EE^{\QQ(a)}\left[e^{\phi v_T+\psi \lambda_T}\right]\Big|\leq \EE^{\QQ(a)}\left[\Big|e^{\phi v_T+\psi \lambda_T}\Big|\right]= \EE^{\QQ(a)}\left[e^{\Re(\phi) v_T+\Re(\psi) \lambda_T}\right],
\end{align*}
and we can apply the same argument as before to conclude that the previous expectation is finite. 

(ii) Since $(N,\lambda)$ is a Markov process, see \cite[Remark 1.22]{ETH}, the process $(v,\lambda)$ is also a Markov and we see that there exists a function $f^{(a)}(\cdot,\cdot,\cdot;\phi,\psi)\colon[0,T]\times(0,\infty)\times[\lambda_0,\infty)\to\RR$ such that 
\begin{align*}
    \EE^{\QQ(a)}\left[e^{\phi v_T+\psi \lambda_T}|\mathcal{F}_t\right] = f^{(a)}(t,v_t,\lambda_t;\phi,\psi).
\end{align*}
Using that $t\to f^{(a)}(t,v_t,\lambda_t;\phi,\psi)$ is a $(\mathcal{F},\QQ(a))$-local martingale and by applying Itô formula similarly as in the previous step we see that it must satisfy the following partial integro-differential equation (PIDE) 
\begin{gather*}
    \partial_tf^{(a)}(t,x,y;\phi,\psi)-\kappa^{(a)}(x-\Bar{v}^{(a)})\partial_xf^{(a)}(t,x,y;\phi,\psi)-\beta(y-\lambda_0)\partial_yf^{(a)}(t,x,y;\phi,\psi) \\+\frac{\sigma^2}{2}x\partial_{xx}^2f^{(a)}(t,x,y;\phi,\psi)
    +y\int_{(0,\infty)}\left[f^{(a)}(t,x+\eta u,y+\alpha;\phi,\psi)-f^{(a)}(t,x,y;\phi,\psi)\right]P_{J_1}(du)=0
\end{gather*}
and final condition $f^{(a)}(T,x,y;\phi,\psi)=e^{\phi x+\psi y}$. Note that the left-hand side of the previous PIDE is just the drift given in \eqref{expr2}. By the affine structure property of the model, see \cite{duffie,VIXJUMPS1}, the solution of the previous PIDE can be written in exponential affine form, that is $f^{(a)}(T,x,y;\phi,\psi)=e^{F^{(a)}(t;\phi,\psi)+G^{(a)}(t;\phi)x+H^{(a)}(t;\phi,\psi)y}$. Substituting that expression in the PIDE we obtain the ODEs given in \eqref{ODEs50} and this finishes the proof.
\end{proof}

\section{Explicit expression of the VIX index}\label{EXPVIX}
In this section we derive an explicit expression of the $\text{VIX}$ index under the Heston-Hawkes model as a linear combination of $v$ and $\lambda$ using the Markov property of the exponential Hawkes. First, a preliminary result is required.

\begin{lem}\label{lemfor}
    Let $\QQ(a)\in\mathcal{E}_m(Q_1,2+\varepsilon_1)$, the following holds
    \begin{enumerate}
        \item The process $t\to\int_0^t\lambda_udX_u^{(a)}$ is a square integrable $(\mathcal{F},\PP)$-martingale. 
        \item Let $0\leq s\leq t\leq T$, then \begin{align*}
            \EE[\lambda_tX_t^{(a)}|\mathcal{F}_s]=\left(\lambda_s-\frac{\beta\lambda_0}{\beta-\alpha}\right)X_s^{(a)}e^{-(\beta-\alpha)(t-s)}+\frac{\beta\lambda_0X_s^{(a)}}{\beta-\alpha}.
        \end{align*}
    \end{enumerate}
\end{lem}
\begin{proof}
    See Lemma \ref{Alemfor} in the appendix. 
\end{proof}

\subsection{Forward variance}
As a starting point, we compute the forward variance under the Heston-Hawkes stochastic volatility model. Since the computations are rather tedious and long the proof is postponed to the appendix.  In order to avoid a singular case in the expression of the forward variance we make the following mild assumption on the family of equivalent martingale measures.

\begin{ass}\label{as4}
     We assume that $a\neq\frac{\beta-\alpha-\kappa}{\sigma}$. This is equivalent to $\kappa^{(a)}\neq\beta-\alpha$ because $\kappa^{(a)}=\kappa+a\sigma$ and ensures that some constants appearing in the next result are well defined. 
\end{ass}

\begin{prop}
Let $\QQ(a)\in\mathcal{E}_m(Q_1,2+\varepsilon_1)$ and define the forward variance by $\xi_{s}^{(a)}(t):=\EE^{\QQ(a)}[v_t|\mathcal{F}_s]$, $0\leq s\leq t\leq T$. Then,
\begin{enumerate}
    \item The explicit expression is given by
\begin{align}\label{expl2}
    \xi_{s}^{(a)}(t)= & \ D_1^{(a)}(t-s)v_s + D_2^{(a)}(t-s)\lambda_s+D_3^{(a)}(t-s),
\end{align}
where
\begin{gather}
D_1^{(a)}(h):=e^{-\kappa^{(a)}h}, \hspace{1cm} D_2^{(a)}(h)=C_1^{(a)}\left(e^{-(\beta-\alpha)h}-e^{-\kappa^{(a)}h}\right), \notag\\ 
D_3^{(a)}(h):=\left(\frac{C_2^{(a)}}{\kappa^{(a)}}-\Bar{v}^{(a)}\right)e^{-\kappa^{(a)}h} -\frac{C_2^{(a)}}{\beta-\alpha}e^{-(\beta-\alpha)h}+C_3^{(a)}, \notag\\
C_1^{(a)}:=\frac{\eta\EE[J_1]}{\kappa^{(a)}-(\beta-\alpha)}, \hspace{1cm} C_2^{(a)}:=\frac{\eta\EE[J_1]\beta\lambda_0}{\kappa^{(a)}-(\beta-\alpha)}, \hspace{1cm} C_3^{(a)}:=\frac{\eta\EE[J_1]\beta\lambda_0}{\kappa^{(a)}(\beta-\alpha)}+\Bar{v}^{(a)}.  \label{eq:const}
\end{gather}

\item The dynamics of the process  $t\to\xi_{s}^{(a)}(t)$ is given by 
\begin{align*}
    d\xi_{s}^{(a)}(t) = \left[-\kappa^{(a)}\left(\xi_{s}^{(a)}(t)-\Bar{v}^{(a)}\right)+\eta\EE[J_1]\left(\left(\lambda_s-\frac{\beta\lambda_0}{\beta-\alpha}\right)e^{-(\beta-\alpha)(t-s)}+\frac{\beta\lambda_0}{\beta-\alpha}\right)\right]dt.
\end{align*}
\item The dynamics of the process $s\to\xi_s(t)$ is given by
       \begin{align*}
        d\xi_s^{(a)}(t)= & \ \sigma D_1^{(a)}(t-s)\sqrt{v_s}dW_s^{\QQ(a)}+\eta D_1^{(a)}(t-s)\left(dL_s-\EE[J_1]\lambda_sds\right) \\ &+\alpha D_2^{(a)}(t-s)\left(dN_s-\lambda_sds\right).
    \end{align*}
\end{enumerate}
\end{prop}
\begin{proof}
  See Proposition \ref{AFOR} in the appendix. 
\end{proof}
\begin{obs}\label{obs1}
    Note that all the constants in \eqref{eq:const} are well defined. By Assumption \ref{as4}, $\kappa^{(a)}-(\beta-\alpha)\neq0$, $\beta-\alpha>0$ due to the stability condition of the Hawkes process and $\kappa^{(a)}>0$ by Theorem \ref{risk}.
\end{obs}

\subsection{Variance Swap}
We now compute an explicit expression of the variance swap using the previous result on the forward variance. The usefulness of this result relies on the fact that the VIX index can be written as a particular variance swap. 
\begin{definition}\label{intdef}
    For $k,h>0$,  we define $A_k(h):=\frac{1}{h}\int_0^he^{-ku}du=\frac{1}{kh}(1-e^{-kh})$.
\end{definition}

Let $0\leq s\leq t\leq T$ and $\QQ(a)\in\mathcal{E}_m(Q_1,2+\varepsilon_1)$, we define the variance swap as
\begin{align*}
    V_s^{(a)}(t):=\frac{1}{t-s}\int_s^t\xi^{(a)}_s(u)du.
\end{align*}
\begin{prop}\label{VIXpre1}
  Let $0\leq s\leq t\leq T$ and $\QQ(a)\in\mathcal{E}_m(Q_1,2+\varepsilon_1)$. Then, the explicit expression of $V_s^{(a)}(t)$ is given by
      \begin{align*}
    V_{s}^{(a)}(t)= & \ K_1^{(a)}(t-s)v_s + K_2^{(a)}(t-s)\lambda_s+K_3^{(a)}(t-s),
\end{align*}
where
\begin{gather}
K_1^{(a)}(h):=A_{\kappa^{(a)}}(h), \hspace{1cm} K_2^{(a)}(h)=C_1^{(a)}\left(A_{\beta-\alpha}(h)-A_{\kappa^{(a)}}(h)\right), \notag\\ 
K_3^{(a)}(h):=\left(\frac{C_2^{(a)}}{\kappa^{(a)}}-\Bar{v}^{(a)}\right)A_{\kappa^{(a)}}(h) -\frac{C_2^{(a)}}{\beta-\alpha}A_{\beta-\alpha}(h)+C_3^{(a)} \label{eq:const2}
\end{gather}
where  $C_1^{(a)},C_2^{(a)}$ and $C_3^{(a)}$ are given in \eqref{eq:const} and $A_k$ defined in Definition \ref{intdef}.
\end{prop}
\begin{proof}
    See Proposition \ref{AlemV} in the appendix. 
\end{proof}

\begin{obs}
    Note that all the constants in \eqref{eq:const2} are well defined by the same argument as in Observation \ref{obs1}, $\kappa^{(a)}-(\beta-\alpha)\neq0$, $\beta-\alpha>0$ and $\kappa^{(a)}>0$.
\end{obs}

\subsection{$\text{VIX}$ index}
To finish this section, we introduce the mathematical definition of the $\text{VIX}$ index and obtain its explicit expression as a linear combination of the variance and the Hawkes intensity. 

\begin{definition}
    Let $\QQ(a)\in\mathcal{E}_m(Q_1,2+\varepsilon_1)$, $t\in[0,T]$ and $\Delta=\frac{30}{365}$. The $\text{VIX}^{(a)}$ introduced in the CBOE (Chicago Board Options Exchange) white paper can be defined, with mathematical simplification, in the following way 
    \begin{align*}
        \left(\text{VIX}_t^{(a)}\right)^2=-\frac{2}{\Delta}\EE^{\QQ(a)}\left[\log\left(e^{-r\Delta}\frac{S_{t+\Delta}}{  S_t}\right)\Big|\mathcal{F}_t\right] \cdot 100^2.
    \end{align*}
    See \cite{VIX1,state1,VIX3,zhulian} for further details on the definition of the VIX index.
\end{definition}

The next lemma will be used to deduce that some constants appearing in the explicit expression of the $\text{VIX}$ index are strictly positive. 
\begin{lem}\label{ineq}
    Let $k_1,k_2,h>0$ with $k_1\neq k_2$. Then, 
    \begin{align*}
        \frac{A_{k_2}(h)-A_{k_1}(h)}{k_1-k_2}>0.
    \end{align*}
\end{lem}
\begin{proof}
  See Lemma \ref{ineqA} in the appendix. 
\end{proof}

\begin{prop}\label{formulaVIX}  Let $t\in[0,T]$ and $\QQ(a)\in\mathcal{E}_m(Q_1,2+\varepsilon_1)$. Then,
    \begin{align*}
        \left(\text{VIX}_t^{(a)}\right)^2=V_t^{(a)}(t+\Delta) \cdot 100^2 = (A^{(a)}v_t+B^{(a)}\lambda_t+C^{(a)})\cdot100^2.
    \end{align*}
    where the constants $A^{(a)},B^{(a)}$ and $C^{(a)}$ are independent of $t$ and are given by
    \begin{align*}
        A^{(a)}&=K_1^{(a)}(\Delta)=A_{\kappa^{(a)}}(\Delta)>0 \\
        B^{(a)}&=K_2^{(a)}(\Delta)=C_1^{(a)}\left(A_{\beta-\alpha}(\Delta)-A_{\kappa^{(a)}}(\Delta)\right)>0,\\
        C^{(a)}&=K_3^{(a)}(\Delta)= \left(\frac{C_2^{(a)}}{\kappa^{(a)}}-\Bar{v}^{(a)}\right)A_{\kappa^{(a)}}(\Delta) -\frac{C_2^{(a)}}{\beta-\alpha}A_{\beta-\alpha}(\Delta)+C_3^{(a)}
    \end{align*}
    where $K_1^{(a)}$, $K_2^{(a)}$ and $K_3^{(a)}$ are given in Proposition \ref{VIXpre1}, $C_1^{(a)}$, $C_2^{(a)}$ and $C_3^{(a)}$ are given in \eqref{eq:const} and $A_k$ defined in Definition \ref{intdef}.
\end{prop}
\begin{proof}
Using that 
    \begin{align*}
        \frac{S_{t+\Delta}}{S_t}=\exp\left(r\Delta-\frac{1}{2}\int_t^{t+\Delta}v_sds+\sqrt{1-\rho^2}\int_{t}^{t+\Delta}\sqrt{v_s}dB_s^{\QQ(a)}+\rho\int_{t}^{t+\Delta}\sqrt{v_s}dW_s^{\QQ(a)}\right)    \end{align*}
        and the fact that both $t\to\int_0^t\sqrt{v_u}dB_u^{\QQ(a)}$ and $t\to\int_0^t\sqrt{v_u}dW_u^{\QQ(a)}$ are $(\mathcal{F},\QQ(a))$-martingales by Theorem \ref{risk}, one can see that the following holds
        \begin{align*}
            \left(\text{VIX}_t^{(a)}\right)^2=V_t^{(a)}(t+\Delta).
        \end{align*}
        Then, the result follows from Proposition \ref{VIXpre1}. Moreover, the fact that $A^{(a)}>0$ follows directly from $A_{\kappa^{(a)}}(\Delta)=\frac{1}{\Delta}\int_0^\Delta e^{-\kappa^{(a)}u}du$ and the fact that $B^{(a)}>0$ follows from the definition of $C_1^{(a)}$ in \eqref{eq:const} and Lemma \ref{ineq}.
\end{proof}

\section{Pricing formula for VIX options}\label{SEC7}
Finally, we have all the results needed to derive a semi-analytical pricing formula for European \text{VIX} call options under the Heston-Hawkes stochastic volatility model. Essentially, the pricing formula is based on the formula in \eqref{fourier}, the explicit expression of $\text{VIX}^2$ and the joint characteristic function of $(v_T,\lambda_T)$. By convenience, we define for $\QQ(a)\in\mathcal{E}_m(Q_1,2+\varepsilon_1)$ and $t\in[0,T]$
\begin{align*}
    \mathcal{C}^{(a)}(t,T,k):=\EE^{\QQ(a)}\left[\max\{\text{VIX}^{(a)}_T/100-k,0\}|\mathcal{F}_t\right]. 
\end{align*}
Then, the price at time $t\in[0,T]$ of an European \text{VIX} call option with maturity time $T$ and strike $K=100K'$ can be written as
\begin{align*}
    e^{-r(T-t)}\EE^{\QQ(a)}\left[\max\{\text{VIX}^{(a)}_T-K,0\}|\mathcal{F}_t\right]=100e^{-r(T-t)}\mathcal{C}^{(a)}\left(t,T,K'\right).
\end{align*}

\begin{thm} Let $\QQ(a)\in\mathcal{E}_m(Q_1,2+\varepsilon_1)$, the price at time $t\in[0,T]$ of an European \text{VIX} call option with maturity time $T$ and strike $K=100K'$  is given by
\begin{align*}
    e^{-r(T-t)}\EE^{\QQ(a)}\left[\max\{\text{VIX}^{(a)}_T-K,0\}|\mathcal{F}_t\right]=100e^{-r(T-t)}\mathcal{C}^{(a)}\left(t,T,K'\right)
\end{align*}
where
\begin{align}\label{formulaprice}
 \mathcal{C}^{(a)}(t,T,k)=\frac{1}{2\sqrt{\pi}}\int_0^\infty\Re\left[\frac{\operatorname{erfc}(k\sqrt{\phi})}{\phi^{3/2}}e^{\phi C^{(a)}}f^{(a)}(t,v_t,\lambda_t;A^{(a)}\phi,B^{(a)}\phi)\right]d\phi_I,
\end{align}
$\text{erfc}$ is the complementary error function defined by $\text{erfc}(z)=1-\text{erf}(z)$, $\phi=\phi_R+i\phi_I\in\mathbb{C}$ satisfies 
\begin{align}\label{phiR}
    0<\phi_R<\min\Bigg\{\frac{2\kappa^{(a)}}{\sigma^2A^{(a)}\left(2e^{\kappa^{(a)}T}-1\right)},\frac{L_J}{A^{(a)}},\frac{\beta-\alpha}{B^{(a)}\alpha\beta}\Bigg\},
\end{align}
$A^{(a)},B^{(a)},C^{(a)}$ are given in Proposition \ref{formulaVIX}, $L_J$ is given in Definition \ref{LJdef} and $f^{(a)}$ is the conditional characteristic function of $(v,\lambda)$ under $\QQ(a)$ given in Proposition \ref{cf}, that is, 
\begin{align*}
    f^{(a)}(t,v_t,\lambda_t;\phi,\psi)=\EE^{\QQ(a)}\left[e^{\phi v_T+\psi \lambda_T}|\mathcal{F}_t\right]=e^{F^{(a)}(t;\phi,\psi)+G^{(a)}(t;\phi)v_t+H^{(a)}(t;\phi,\psi)\lambda_t},
\end{align*}
where $F^{(a)}, G^{(a)}$ and $H^{(a)}$ satisfy the ODEs given in \eqref{ODEs50}, that is,
\begin{align*}
    \frac{d}{dt}G^{(a)}(t;\phi)-\kappa^{(a)}G^{(a)}(t;\phi)+\frac{1}{2}\sigma^2G^{(a)}(t;\phi)^2=0,& \hspace{0.5cm}G^{(a)}(T;\phi)=\phi, \notag\\
    \frac{d}{dt}H^{(a)}(t;\phi,\psi)-\beta H^{(a)}(t;\phi,\psi)+e^{\alpha H^{(a)}(t;\phi,\psi)}\EE\left[e^{\eta G^{(a)}(t;\phi) J_1}\right]-1=0,& \hspace{0.5cm}H^{(a)}(T;\phi,\psi)=\psi, \notag\\
    \frac{d}{dt}F^{(a)}(t;\phi,\psi)+\kappa^{(a)}\Bar{v}^{(a)}G^{(a)}(t;\phi)+\beta\lambda_0H^{(a)}(t;\phi,\psi)=0,& \hspace{0.5cm}F^{(a)}(T;\phi,\psi)=0.
\end{align*}
\end{thm} 
\begin{proof}
To obtain the pricing formula in \eqref{formulaprice} we use the expression of $\text{VIX}^2$ given in Proposition \ref{formulaVIX}, the formula in \eqref{fourier} and the joint conditional characteristic function of $(v_T,\lambda_T)$ given in Proposition \ref{cf}:
\begin{align*}
    \mathcal{C}^{(a)}(t,T,k) & =\EE^{\QQ(a)}\left[\max\{\text{VIX}^{(a)}_T/100-k,0\}|\mathcal{F}_t\right] \\
    & = \EE^{\QQ(a)}\left[\max\{\sqrt{A^{(a)}v_T+B^{(a)}\lambda_T+C^{(a)}}-k,0\}|\mathcal{F}_t\right] \\
    & = \frac{1}{2\sqrt{\pi}}\int_0^\infty\Re\left[\frac{\text{erfc}\left(k\sqrt{\phi}\right)}{\phi^{3/2}}\EE^{\QQ(a)}\left[e^{\phi\left(A^{(a)}v_T+B^{(a)}\lambda_T+C^{(a)}\right)}\big|\mathcal{F}_t\right]\right]d\phi_I \\
    & = \frac{1}{2\sqrt{\pi}}\int_0^\infty\Re\left[\frac{\text{erfc}(k\sqrt{\phi})}{\phi^{3/2}}e^{\phi C^{(a)}}f^{(a)}(t,v_t,\lambda_t;A^{(a)}\phi,B^{(a)}\phi)\right]d\phi_I.
\end{align*}
Note that the inequality in \eqref{phiR} ensures that the characteristic function is well defined, see Proposition \ref{cf}. Moreover, we have used that $A^{(a)},B^{(a)}>0$ which is proved in Proposition \ref{formulaVIX} and the fact that $\phi_R>0$ guarantees that we can use the formula in \eqref{fourier}.
\end{proof}

\newpage
\renewcommand{\abstractname}{Acknowledgements}
\begin{abstract}
The author would like to thank David Ruiz Banos and Salvador Ortiz-Latorre for their help and insightful feedback during the development of this article.
\end{abstract}

\appendix
\section{Appendix: Technical results}\label{sec: appendix}
We give the proofs that were postponed.

\subsection{Joint characteristic function}
\begin{lem}\label{AJQP} Let $\QQ(a)\in\mathcal{E}_m(Q_1,2+\varepsilon_1)$, then $P_{J_1}=\QQ(a)_{J_1}$, that is, the law of $J_1$ is the same under $\PP$ and under $\QQ(a)$. 
\end{lem}
\begin{proof}
   Let $f$ be a measurable and bounded function,  we are going to prove that $\EE^{\QQ(a)}[f(J_1)]=\EE[f(J_1)]$ which will imply that the law of $J_1$ is the same under $\PP$ and $\QQ(a)$. By Theorem \ref{risk}, $\frac{d\QQ(a)}{d\PP}=X_T^{(a)}=Y_T^{(a)}Z_T^{(a)}$. Now, using that $Z_T^{(a)}$ is $\mathcal{F}_T^W\vee\mathcal{F}_T^L$-measurable, that
   \begin{align*}
       Y_T^{(a)}|\mathcal{F}_T^W\vee\mathcal{F}_T^L\vee\mathcal{F}^{J_1}\sim\text{Lognormal}\left(-\frac{1}{2}\int_0^T\left(\theta_s^{(a)}\right)^2ds,\int_0^T\left(\theta_s^{(a)}\right)^2ds\right)
   \end{align*}
   because $\theta^{(a)}$ is $\{\mathcal{F}_t^W\vee\mathcal{F}_t^L\vee\mathcal{F}^{J_1}\}_{t\in[0,T]}$-adapted and the fact that $Z^{(a)}$ is a $\{\mathcal{F}_t^W\vee\mathcal{F}_T^L\vee\mathcal{F}^{J_1}\}_{t\in[0,T]}$-martingale we have 
   \begin{align*}
       \EE^{\QQ(a)}[f(J_1)]&=\EE[f(J_1)X_T^{(a)}]=\EE[f(J_1)Y_T^{(a)}Z_T^{(a)}]=\EE[\EE[f(J_1)Y_T^{(a)}Z_T^{(a)}|\mathcal{F}_T^W\vee\mathcal{F}_T^L\vee\mathcal{F}^{J_1}]] \\
       &=\EE[f(J_1)Z_T^{(a)}\EE[Y_T^{(a)}|\mathcal{F}_T^W\vee\mathcal{F}_T^L\vee\mathcal{F}^{J_1}]]= \EE[f(J_1)Z_T^{(a)}] \\
       &=\EE[\EE[f(J_1)Z_T^{(a)}|\mathcal{F}_0^W\vee\mathcal{F}_T^L\vee\mathcal{F}^{J_1}]]=\EE[f(J_1)\EE[Z_T^{(a)}|\mathcal{F}_0^W\vee\mathcal{F}_T^L\vee\mathcal{F}^{J_1}]]=\EE[f(J_1)].
   \end{align*}
\end{proof}
\begin{lem}\label{AODEs} Let $\QQ(a)\in\mathcal{E}_m(Q_1,2+\varepsilon_1)$ and $\phi,\psi\in\mathbb{C}$ such that
\begin{align*}
    0<\Re(\phi)<\min\Bigg\{\frac{2\kappa^{(a)}}{\sigma^2\left(2e^{\kappa^{(a)}T}-1\right)},L_J\Bigg\} \hspace{1cm}\text{and}\hspace{1cm} \Re(\psi)<\frac{\beta-\alpha}{\alpha\beta}.
\end{align*}
Then,
\begin{enumerate}[(i)]
    \item The ODE 
    \begin{align*}
        \frac{d}{dt}G^{(a)}(t;\phi)-\kappa^{(a)}G^{(a)}(t;\phi)+\frac{1}{2}\sigma^2G^{(a)}(t;\phi)^2=0, \hspace{0.5cm} G^{(a)}(T;\phi)=\phi,
    \end{align*}
    has a unique solution in the interval $[0,T]$ given by
    \begin{align*}
        G^{(a)}(t;\phi)=\frac{2\kappa^{(a)}}{\sigma^2+e^{\kappa^{(a)}(T-t)}\left(\frac{2\kappa^{(a)}}{\phi}-\sigma^2\right)}.
    \end{align*}
    \item $\sup_{t\in[0,T]}\Re(G^{(a)}(t;\phi))=\Re(\phi)$.
    \item $\sup_{t\in[0,T]}\EE\left[e^{\eta \Re(G^{(a)}(t;\phi)) J_1}\right]<\frac{\beta}{\alpha}\exp\left(\frac{\alpha}{\beta}-1\right)$. In particular, $\Big|\EE\left[e^{\eta G^{(a)}(t;\phi) J_1}\right]\Big|<\infty$ for all $t\in[0,T]$. 
    \item The ODE 
    \begin{align}\label{ODEHAPP}
         \frac{d}{dt}H^{(a)}(t;\phi,\psi)-\beta H^{(a)}(t;\phi,\psi)+e^{\alpha H^{(a)}(t;\phi,\psi)}\EE\left[e^{\eta G^{(a)}(t;\phi) J_1}\right]-1=0, \hspace{0.5cm} H^{(a)}(T;\phi,\psi)=\psi
    \end{align}
    has a unique solution in $[0,T]$. 
   \item The ODE
   \begin{align*}
        \frac{d}{dt}F^{(a)}(t;\phi,\psi)+\kappa^{(a)}\Bar{v}^{(a)}G^{(a)}(t;\phi)+\beta\lambda_0H^{(a)}(t;\phi,\psi)=0, \hspace{0.5cm} F^{(a)}(T;\phi,\psi)=0.
   \end{align*}
   has a unique solution in $[0,T]$. 
\end{enumerate}
\end{lem}
\begin{proof}
(i) One can check that the solution of the ODE is 
\begin{align*}
    G^{(a)}(t;\phi)=\frac{2\kappa^{(a)}}{\sigma^2+e^{\kappa^{(a)}(T-t)}\left(\frac{2\kappa^{(a)}}{\phi}-\sigma^2\right)}.
\end{align*}
Note that the denominator of $G^{(a)}(t;\phi)$ is $0$ for $t\in[0,T)$ if and only if
\begin{align*}
    \phi=\frac{2\kappa^{(a)}}{\sigma^2}\frac{e^{\kappa^{(a)}(T-t)}}{e^{\kappa^{(a)}(T-t)}-1}.
\end{align*}
A necessary condition for the previous equality to hold is that $\Im(\phi)=0$. It can be checked that 
\begin{align*}
    \inf_{t\in[0,T]}\frac{2\kappa^{(a)}}{\sigma^2}\frac{e^{\kappa^{(a)}(T-t)}}{e^{\kappa^{(a)}(T-t)}-1}= \frac{2\kappa^{(a)}}{\sigma^2}\frac{e^{\kappa^{(a)}T}}{e^{\kappa^{(a)}T}-1},
\end{align*}
and since $\Re(\phi)<\frac{2\kappa^{(a)}}{\sigma^2\left(2e^{\kappa^{(a)}T}-1\right)}<\frac{2\kappa^{(a)}}{\sigma^2}\frac{e^{\kappa^{(a)}T}}{e^{\kappa^{(a)}T}-1}\leq\frac{2\kappa^{(a)}}{\sigma^2}\frac{e^{\kappa^{(a)}(T-t)}}{e^{\kappa^{(a)}(T-t)}-1}$ for all $t\in[0,T)$, $ G^{(a)}$ is well defined in the whole interval $[0,T]$. 

(ii) First, we compute the real part of $ G^{(a)}$.  Writing $DG^{(a)}(t;\phi)=\sigma^2+e^{\kappa^{(a)}(T-t)}\left(\frac{2\kappa^{(a)}}{\phi}-\sigma^2\right)$ we have
   \begin{align*}
       \Re(G^{(a)}(t;\phi))=\frac{2\kappa^{(a)}\Re(DG^{(a)}(t;\phi))}{|DG^{(a)}(t;\phi)|^2}.
   \end{align*}
   Then,
   \begin{align*}
       \Re(DG^{(a)}(t;\phi)) &=\sigma^2+e^{\kappa^{(a)}(T-t)}\left(\frac{2\kappa^{(a)}\Re(\phi)}{|\phi|^2}-\sigma^2\right) \\
       &=\sigma^2\left(1-e^{\kappa^{(a)}(T-t)}\right)+\frac{2\kappa^{(a)}\Re(\phi)}{|\phi|^2}e^{\kappa^{(a)}(T-t)}, \\
       \Im(DG^{(a)}(t;\phi)) &=-\frac{2\kappa^{(a)}\Im(\phi)}{|\phi|^2}e^{\kappa^{(a)}(T-t)}.
   \end{align*}
   Note that
   \begin{align*}
       |DG^{(a)}(t;\phi)|^2 = \ & \sigma^4\left(1-e^{\kappa^{(a)}(T-t)}\right)^2+\frac{4\left(\kappa^{(a)}\right)^2\Re(\phi)^2}{|\phi|^4}e^{2\kappa^{(a)}(T-t)} \\ &+\frac{4\kappa^{(a)}\sigma^2\left(1-e^{\kappa^{(a)}(T-t)}\right)e^{\kappa^{(a)}(T-t)}\Re(\phi)}{|\phi|^2} 
       +\frac{4\left(\kappa^{(a)}\right)^2\Im(\phi)^2}{|\phi|^4}e^{2\kappa^{(a)}(T-t)} \\
       = \ &  \sigma^4\left(1-e^{\kappa^{(a)}(T-t)}\right)^2+\frac{4\left(\kappa^{(a)}\right)^2}{|\phi|^2}e^{2\kappa^{(a)}(T-t)}+\frac{4\kappa^{(a)}\sigma^2\left(1-e^{\kappa^{(a)}(T-t)}\right)e^{\kappa^{(a)}(T-t)}\Re(\phi)}{|\phi|^2}.
   \end{align*}
    Thus, 
    \begin{align*}
        \Re(G^{(a)}(t;\phi))=\frac{2\kappa^{(a)}\sigma^2\left(1-e^{\kappa^{(a)}(T-t)}\right)|\phi|^2+4\left(\kappa^{(a)}\right)^2e^{\kappa^{(a)}(T-t)}\Re(\phi)}{\sigma^4\left(1-e^{\kappa^{(a)}(T-t)}\right)^2|\phi|^2+4\left(\kappa^{(a)}\right)^2e^{2\kappa^{(a)}(T-t)}+4\kappa^{(a)}\sigma^2\left(1-e^{\kappa^{(a)}(T-t)}\right)e^{\kappa^{(a)}(T-t)}\Re(\phi)}.
    \end{align*}
We check that the denominator of $\Re(G^{(a)}(t;\phi))$ in the previous expression is always strictly positive for $t\in[0,T)$. If $t\in[0,T)$ we have the following
\begin{align*}
    \sigma^4\left(1-e^{\kappa^{(a)}(T-t)}\right)^2|\phi|^2+4\left(\kappa^{(a)}\right)^2e^{2\kappa^{(a)}(T-t)}+4\kappa^{(a)}\sigma^2\left(1-e^{\kappa^{(a)}(T-t)}\right)e^{\kappa^{(a)}(T-t)}\Re(\phi) \\
    \geq\sigma^4\left(1-e^{\kappa^{(a)}(T-t)}\right)^2\Re(\phi)^2+4\left(\kappa^{(a)}\right)^2e^{2\kappa^{(a)}(T-t)}+4\kappa^{(a)}\sigma^2\left(1-e^{\kappa^{(a)}(T-t)}\right)e^{\kappa^{(a)}(T-t)}\Re(\phi) \\
    = f_t(\Re(\phi)),
\end{align*}
where $f_t$ is defined by
\begin{align*}
    f_t(x):=\sigma^4\left(1-e^{\kappa^{(a)}(T-t)}\right)^2x^2+4\left(\kappa^{(a)}\right)^2e^{2\kappa^{(a)}(T-t)}+4\kappa^{(a)}\sigma^2\left(1-e^{\kappa^{(a)}(T-t)}\right)e^{\kappa^{(a)}(T-t)}x.
\end{align*}
Using that $f_t$ is a quadratic function in $x$, one can check that the minimum of $f_t$ is $0$ and it is achieved at $x_t=\frac{2\kappa^{(a)}}{\sigma^2}\frac{e^{\kappa^{(a)}(T-t)}}{e^{\kappa^{(a)}(T-t)}-1}>0$ and for other values of $x$, $f_t(x)$ is strictly positive. Using again that
\begin{align*}
    \inf_{t\in[0,T)}\frac{2\kappa^{(a)}}{\sigma^2}\frac{e^{\kappa^{(a)}(T-t)}}{e^{\kappa^{(a)}(T-t)}-1}=\frac{2\kappa^{(a)}}{\sigma^2}\frac{e^{\kappa^{(a)}T}}{e^{\kappa^{(a)}T}-1},
\end{align*}
and that $\Re(\phi)<\frac{2\kappa^{(a)}}{\sigma^2(2e^{\kappa^{(a)}T}-1)}<\frac{2\kappa^{(a)}}{\sigma^2}\frac{e^{\kappa^{(a)}T}}{e^{\kappa^{(a)}T}-1}\leq\frac{2\kappa^{(a)}}{\sigma^2}\frac{e^{\kappa^{(a)}(T-t)}}{e^{\kappa^{(a)}(T-t)}-1}$ for all $t\in[0,T)$, we conclude that the denominator of $\Re(G^{(a)}(t;\phi))$  is strictly positive for all $t\in[0,T]$.

Finally, we prove that $0<\Re(\phi)<\frac{2\kappa^{(a)}}{\sigma^2(2e^{\kappa^{(a)}T}-1)}$ implies that $\sup_{t\in[0,T]}\Re(G^{(a)}(t;\phi))=\Re(\phi)$. Recall that $\Re(G^{(a)}(T;\phi))=\Re(\phi)$ and we need to prove that $\Re(G^{(a)}(t;\phi))\leq\Re(\phi)$ for $t\in[0,T)$. Using that the denominator is always positive one can check that  $\Re(G^{(a)}(t;\phi))\leq\Re(\phi)$ is equivalent to 

\begin{align}\label{ineq1}
    L_1(t)\Re(\phi)^3+L_2(t)\Re(\phi)^2+L_3(t)\Re(\phi)+L_4(t)\Im(\phi)^2+L_5(t)\Im(\phi)^2\Re(\phi)\leq0,
\end{align}
where
\begin{align*}
    L_1(t)&=-\sigma^4\left(1-e^{\kappa^{(a)}(T-t)}\right)^2, \\
    L_2(t)&=2\kappa^{(a)}\sigma^2\left(1-e^{\kappa^{(a)}(T-t)}\right)\left(1-2e^{\kappa^{(a)}(T-t)}\right), \\
    L_3(t)&=4\left(\kappa^{(a)}\right)^2e^{\kappa^{(a)}(T-t)}\left(1-e^{\kappa^{(a)}(T-t)}\right), \\
    L_4(t)&=2\kappa^{(a)}\sigma^2\left(1-e^{\kappa^{(a)}(T-t)}\right), \\
    L_5(t)&=-\sigma^4\left(1-e^{\kappa^{(a)}(T-t)}\right)^2.
\end{align*}
Using that $L_4(t)<0$ and $L_5(t)<0$ for all $t\in[0,T)$, we obtain the following
\begin{align*}
    L_1(t)\Re(\phi)^3+L_2(t)\Re(\phi)^2+L_3(t)\Re(\phi)+L_4(t)\Im(\phi)^2+L_5(t)\Im(\phi)^2\Re(\phi)  \\
    \leq L_1(t)\Re(\phi)^3+L_2(t)\Re(\phi)^2+L_3(t)\Re(\phi)\leq 0.
\end{align*}
Using that $\Re(\phi)>0$ and that $\left(1-e^{\kappa^{(a)}(T-t)}\right)<0$ for $t\in[0,T)$, the inequality in \eqref{ineq1} is equivalent to
\begin{align*}
    \widetilde{L}_1(t)\Re(\phi)^2+\widetilde{L}_2(t)\Re(\phi)+\widetilde{L}_3(t)\leq 0,
\end{align*}
where
\begin{align*}
    \widetilde{L}_1(t)&=\sigma^4\left(1-e^{\kappa^{(a)}(T-t)}\right), \\
    \widetilde{L}_2(t)&=2\kappa^{(a)}\sigma^2\left(2e^{\kappa^{(a)}(T-t)}-1\right), \\
    \widetilde{L}_3(t)&=-4\left(\kappa^{(a)}\right)^2e^{\kappa^{(a)}(T-t)}.
\end{align*}
Note that
\begin{align*}
    \widetilde{L}_1(t)\Re(\phi)^2+\widetilde{L}_2(t)\Re(\phi)+\widetilde{L}_3(t)&\leq \sup_{t\in[0,T]}\widetilde{L}_1(t)\Re(\phi)^2+\sup_{t\in[0,T}\widetilde{L}_2(t)\Re(\phi)+\sup_{t\in[0,T]}\widetilde{L}_3(t) \\
    &\leq 2\kappa^{(a)}\sigma^2\left(2e^{\kappa^{(a)}T}-1\right)\Re(\phi)-4\left(\kappa^{(a)}\right)^2\leq 0
\end{align*}
which holds because $0<\Re(\phi)<\frac{2\kappa^{(a)}}{\sigma^2(2e^{\kappa^{(a)}T}-1)}$. This finishes the proof. 

(iii) Using that $L_J=\frac{1}{\eta}M_J^{-1}\left(\frac{\beta}{\alpha}\exp\left(\frac{\alpha}{\beta}-1\right)\right)$ and that $\Re(\phi)<L_J$ we have
\begin{align*}
    \sup_{t\in[0,T]}\EE\left[e^{\eta\Re(G^{(a)}(t;\phi))J_1}\right]&\leq\EE\left[\sup_{t\in[0,T]}e^{\eta\Re(G^{(a)}(t;\phi))J_1}\right]=\EE\left[e^{\eta\sup_{t\in[0,T]}\Re(G^{(a)}(t;\phi))J_1}\right] \\
    &=\EE\left[e^{\eta\Re(\phi)J_1]}\right]<\EE\left[e^{\eta L_JJ_1}\right]=\frac{\beta}{\alpha}\exp\left(\frac{\alpha}{\beta}-1\right).
\end{align*}
In particular, 
 \begin{align*}
        \Big|\EE\left[e^{\eta G^{(a)}(t;\phi) J_1}\right]\Big|&\leq\EE\left[\Big| e^{\eta G^{(a)}(t;\phi) J_1}\Big|\right]=\EE\left[e^{\eta \Re(G^{(a)}(t;\phi)) J_1}\right] \\
        &\leq\sup_{t\in[0,T]}\EE\left[e^{\eta\Re(G^{(a)}(t;\phi))J_1}\right]<\frac{\beta}{\alpha}\exp\left(\frac{\alpha}{\beta}-1\right)<\infty,
   \end{align*}
   for all $t\in[0,T]$.

(iv) We introduce the function $h^{(a)}(t;\phi,\psi):=H^{(a)}(T-t;\phi,\psi)$. Then, the ODE for $h^{(a)}$ is 
\begin{align*}
    \frac{d}{dt}h^{(a)}(t;\phi,\psi)&=f^{(a)}(t,h^{(a)}(t;\phi,\psi)) \\
    &=e^{\alpha h^{(a)}(t;\phi,\psi)}\EE\left[e^{\eta G^{(a)}(T-t;\phi) J_1}\right]-\beta h^{(a)}(t;\phi,\psi)-1, \hspace{1cm} h^{(a)}(0;\phi,\psi)=\psi,
\end{align*}
where $f^{(a)}(t,x)=e^{\alpha x}\EE\left[e^{\eta G^{(a)}(T-t;\phi) J_1}\right]-\beta x-1\in\mathbb{C}$. Since $f^{(a)}$ is a continuously differentiable function, it is Lipschitz continuous on bounded intervals and there exists a unique local solution for every initial condition, see \cite[Chapter II, Theorem 1.1]{hartman}. First, we study the imaginary part of $h^{(a)}$. Define $U^{(a)}:=\sup_{t\in[0,T]}\EE\left[e^{\eta\Re(G^{(a)}(t;\phi))J_1}\right]<\infty$ and note that
\begin{align*}
    \frac{d}{dt}\Im(h^{(a)}(t;\phi,\psi))&\leq\Big|e^{\alpha h^{(a)}(t;\phi,\psi)}\EE\left[e^{\eta G^{(a)}(T-t;\phi) J_1}\right]\Big|-\beta\Im(h^{(a)}(t;\phi,\psi)) \\
    &\leq e^{\alpha\Re(h^{(a)}(t;\phi,\psi))}\EE\left[e^{\eta\Re(G^{(a)}(T-t;\phi))J_1}\right]-\beta\Im(h^{(a)}(t;\phi,\psi)) \\
    &\leq e^{\alpha\sup_{t\in[0,T]}\Re(h^{(a)}(t;\phi,\psi))}U^{(a)}-\beta\Im(h^{(a)}(t;\phi,\psi)) \\
    &= C^{(a)}-\beta\Im(h^{(a)}(t;\phi,\psi)),
\end{align*}
where $C^{(a)}:=e^{\alpha\sup_{t\in[0,T]}\Re(h^{(a)}(t))}U^{(a)}$. Similarly, we can get the following lower bound for $\frac{d}{dt}\Im(h^{(a)}(t;\phi,\psi))$:
\begin{align*}
    -C^{(a)}-\beta\Im(h^{(a)}(t;\phi,\psi)))\leq\frac{d}{dt}\Im(h^{(a)}(t;\phi,\psi))).
\end{align*}
Therefore, 
\begin{align*}
     -C^{(a)}-\beta\Im(h^{(a)}(t;\phi,\psi)))\leq\frac{d}{dt}\Im(h^{(a)}(t;\phi,\psi)))\leq C^{(a)}-\beta\Im(h^{(a)}(t;\phi,\psi))).
\end{align*}
Define now the two following ODEs
\begin{align*}
    \frac{d}{dt}h^{(a)}_{IM}(t;\phi,\psi))=C^{(a)}-\beta h^{(a)}_{IM}(t;\phi,\psi))&, \hspace{1cm} h^{(a)}_{IM}(0;\phi,\psi))=\Im(\psi), \\
    \frac{d}{dt}h^{(a)}_{Im}(t;\phi,\psi))=-C^{(a)}-\beta h^{(a)}_{Im}(t;\phi,\psi))&, \hspace{1cm} h^{(a)}_{Im}(0;\phi,\psi))=\Im(\psi).
\end{align*}
Then, by the comparison theorem 
\begin{align*}
    -\frac{C^{(a)}}{\beta}\left(1-e^{-\beta t}\right)+e^{-\beta t}\Im(\psi)=h^{(a)}_{Im}(t;\phi,\psi))\leq\Im(h^{(a)}(t;\phi,\psi)))
\end{align*}
and
\begin{align*}
    \Im(h^{(a)}(t;\phi,\psi)))\leq h^{(a)}_{IM}(t;\phi,\psi)) = \frac{C^{(a)}}{\beta}\left(1-e^{-\beta t}\right)+e^{-\beta t}\Im(\psi)
\end{align*}
for all $t\in[0,T]$. This proves that as long as $C^{(a)}<\infty$, the imaginary part of $h^{(a)}$ does not blow up in $[0,T]$. Since $C^{(a)}=e^{\alpha\sup_{t\in[0,T]}\Re(h^{(a)}(t))}U^{(a)}$ and $U^{(a)}<\infty$, it suffices to check that the real part of $h^{(a)}$ does not explode in $[0,T]$ to ensure existence of $h^{(a)}$ over the whole interval $[0,T]$. 

Similarly as done with the imaginary part we can get the following bounds on the derivative of the real part of $h^{(a)}$:
\begin{align*}
    -e^{\alpha\Re(h^{(a)}(t;\phi,\psi))}U^{(a)}-\beta\Re(h^{(a)}(t;\phi,\psi))-1\leq\frac{d}{dt}\Re(h^{(a)}(t;\phi,\psi))
    \end{align*} and \begin{align*}
    \frac{d}{dt}\Re(h^{(a)}(t;\phi,\psi))\leq e^{\alpha\Re(h^{(a)}(t;\phi,\psi))}U^{(a)}-\beta\Re(h^{(a)}(t;\phi,\psi))-1.
\end{align*}
Define now the two following ODEs
\begin{align*}
    \frac{d}{dt}h^{(a)}_{RM}(t;\phi,\psi)&=f^{(a)}_{M}(h^{(a)}_{RM}(t;\phi,\psi)) \\
    &=U^{(a)}e^{\alpha h^{(a)}_{RM}(t;\phi,\psi)}-\beta h^{(a)}_{RM}(t;\phi,\psi)-1, \hspace{1cm} h^{(a)}_{RM}(0;\phi,\psi)=\Re(\psi), \\
    \frac{d}{dt}h^{(a)}_{Rm}(t;\phi,\psi)&=f^{(a)}_{m}(h^{(a)}_{Rm}(t;\phi,\psi))  \\
    &=-U^{(a)}e^{\alpha h^{(a)}_{Rm}(t;\phi,\psi)}-\beta h^{(a)}_{Rm}(t;\phi,\psi)-1, \hspace{1cm} h^{(a)}_{Rm}(0;\phi,\psi)=\Re(\psi),
\end{align*}
where $f^{(a)}_{M}(x)=U^{(a)}e^{\alpha x}-\beta x-1$ and $f^{(a)}_{m}(x)=-U^{(a)}e^{\alpha x}-\beta x-1$. Then, by the comparison theorem and as long as $h^{(a)}_{RM}$ and $h^{(a)}_{Rm}$ exist over the whole interval $[0,T]$ we have
\begin{align}\label{comp}
    h^{(a)}_{Rm}(t;\phi,\psi)\leq \Re(h^{(a)}(t;\phi,\psi)) \leq h^{(a)}_{RM}(t;\phi,\psi),
\end{align}
for all $t\in[0,T]$. We study the existence of $h^{(a)}_{RM}$ in the interval $[0,T]$. One can check that $f^{(a)}_M$ has a unique absolute minimum at $x^{(a)}_{\text{min}}=\frac{1}{\alpha}\log\left(\frac{\beta}{\alpha U^{(a)}}\right)$ and it is equal to
\begin{align*}
    f^{(a)}_M(x_{\text{min}})=\frac{\beta}{\alpha}\left(1-\log\left(\frac{\beta}{\alpha U}\right)\right)-1.
\end{align*}
Note that $f^{(a)}_M(x_{\text{min}})<0$ if and only if $U^{(a)}<\frac{\beta}{\alpha}\exp\left(\frac{\alpha}{\beta}-1\right)$, which holds by (iii). Moreover, 
\begin{align}\label{ineq2}
    x^{(a)}_{\text{min}}=\frac{1}{\alpha}\log\left(\frac{\beta}{\alpha U^{(a)}}\right)>\frac{\beta-\alpha}{\alpha\beta}>0. 
\end{align}
Since $\lim_{x\to\pm\infty}f^{(a)}_M(x)=\infty$ and there is no local minimum or maximum, there exist exactly two zeroes $x^{(a)}_{z1}$ and $x^{(a)}_{z2}$ of $f^{(a)}_M$ that satisfy $x^{(a)}_{z_1}<x^{(a)}_{\text{min}}<x^{(a)}_{z2}$. Note that $x^{(a)}_{z1}$ is a stable equilibrium point and for all initial conditions in the interval $(-\infty,x^{(a)}_{z2})$, $h_{RM}$ exists for all $t\in[0,\infty)$ and it converges to $x^{(a)}_{z1}$. Therefore, using \eqref{ineq2} and $\Re(\psi)<\frac{\beta-\alpha}{\alpha\beta}$ we have
\begin{align*}
    \Re(\psi)<\frac{\beta-\alpha}{\alpha\beta}<\frac{1}{\alpha}\log\left(\frac{\beta}{\alpha U^{(a)}}\right)=x^{(a)}_{\text{min}}<x^{(a)}_{z2}
\end{align*}
We conclude that $h^{(a)}_{RM}$ exists for all $t\in[0,\infty)$ and it converges to $x^{(a)}_{z_1}$. 

We study now the existence of $h^{(a)}_{Rm}$ in the interval $[0,T]$. One can check that $f^{(a)}_m$ is decreasing in $\RR$ and that $\lim_{x\to-\infty}f^{(a)}_m(x)=+\infty$ and $\lim_{x\to\infty}f^{(a)}_m(x)=-\infty$. This implies the existence of a unique zero which is a stable equilibrium point, and for any initial condition the solution of the ODE exists for all $t\in[0,\infty)$ and it converges to that zero of $f_m$. Therefore, there is no constraint on $\Re(\psi)$ to guarantee the existence of $h^{(a)}_{Rm}$ over the interval $[0,T]$. 

Finally, by \eqref{comp} we conclude that $\Re(h^{(a)}(t))$ exists for all $t\in[0,T]$ and $\Im(h^{(a)}(t))$ as well. Hence, the ODE in \eqref{ODEHAPP} has a unique solution in $[0,T]$. 

(v) Since $G^{(a)}$ and $H^{(a)}$ exist for all $t\in[0,T]$, $F^{(a)}$ also does. 
\end{proof}

\subsection{Explicit expression of the VIX index}

\begin{lem}\label{Alemfor}
    Let $\QQ(a)\in\mathcal{E}_m(Q_1,2+\varepsilon_1)$, the following holds
    \begin{enumerate}
        \item The process $t\to\int_0^t\lambda_udX_u^{(a)}$ is a square integrable $(\mathcal{F},\PP)$-martingale. 
        \item Let $0\leq s\leq t\leq T$, then \begin{align*}
            \EE[\lambda_tX_t^{(a)}|\mathcal{F}_s]=\left(\lambda_s-\frac{\beta\lambda_0}{\beta-\alpha}\right)X_s^{(a)}e^{-(\beta-\alpha)(t-s)}+\frac{\beta\lambda_0X_s^{(a)}}{\beta-\alpha}.
        \end{align*}
    \end{enumerate}
\end{lem}
\begin{proof}
    (1) By Theorem \ref{risk} and Proposition \ref{oldprop}(1) we know that $X^{(a)}$ is a square integrable martingale. We check that
\begin{align*}
    \EE\left[\int_0^T\lambda_{u-}^2d[X^{(a)}]_u\right]<\infty
\end{align*}
to prove that $t\to\int_0^t\lambda_{u-}dX_u^{(a)}$ is a square integrable $(\mathcal{F},\PP)$-martingale. The quadratic variation of $X^{(a)}$ is given by $d[X^{(a)}]_t=\left[\left(\theta_t^{(a)}\right)^2+a^2v_t\right]\left(X_t^{(a)}\right)^2dt$, where $\theta^{(a)}$ is defined in Theorem \ref{risk}. Then, 
\begin{align}\label{Lco}
    \EE\left[\int_0^T\lambda_{u-}^2d[X^{(a)}]_u\right] & = \EE\left[\int_0^T\lambda_{u}^2\left[\left(\theta_u^{(a)}\right)^2+a^2v_u\right]\left(X_u^{(a)}\right)^2du\right]\notag \\ 
    & =\int_0^T\EE\left[\lambda_{u}^2\left(\theta_u^{(a)}\right)^2\left(X_u^{(a)}\right)^2\right]du+a^2\int_0^T\EE\left[\lambda_{u}^2v_u\left(X_u^{(a)}\right)^2\right]du.
\end{align}
We focus on the first term in \eqref{Lco}. Recall that $\varepsilon_1$ and $\varepsilon_2$ where fixed in Assumption \ref{as3} and their connection with the results in Proposition \ref{oldprop}. Applying Hölder's inequality with $p_1=\frac{p_2p_3}{p_2p_3-p_2-p_3}>1$, $p_2=1+\varepsilon_2>1$, $p_3=1+\frac{\varepsilon_1}{2}>1$,  and then Doob's martingale inequality, see \cite[Section 2.1.2, Theorem 2.1.5]{applebaum_2009}, to the last expectation we get
\begin{align}\label{Lcota}
    \EE\left[\lambda_{u}^2\left(\theta_u^{(a)}\right)^2\left(X_u^{(a)}\right)^2\right]&\leq\EE[\lambda_{u}^{2p_1}]^{\frac{1}{p_1}}\EE\left[\left(\theta_u^{(a)}\right)^{2p_2}\right]^{\frac{1}{p_2}}\EE\left[\left(X_u^{(a)}\right)^{2p_3}\right]^{\frac{1}{p_3}} \notag\\
    & = \EE[\lambda_{u}^{2p_1}]^{\frac{1}{p_1}}\EE\left[\left(\theta_u^{(a)}\right)^{2+2\varepsilon_2}\right]^{\frac{1}{p_2}}\EE\left[\left(X_u^{(a)}\right)^{2+\varepsilon_1}\right]^{\frac{1}{p_3}} \notag\\
    & \leq \left(\frac{2+\varepsilon_1}{1+\varepsilon_1}\right)^{2}\EE[\lambda_{u}^{2p_1}]^{\frac{1}{p_1}}\EE\left[\left(\theta_u^{(a)}\right)^{2+2\varepsilon_2}\right]^{\frac{1}{p_2}}\EE\left[\left(X_T^{(a)}\right)^{2+\varepsilon_1}\right]^{\frac{1}{p_3}}.
\end{align}
By \cite[Section 3.1]{momentshawkes}, $\EE[\lambda_{u}^{2p_1}]<\infty$ for all $u\in[0,T]$ and it is continuous as a function of $u$. Moreover, by Proposition \ref{oldprop}(1) we have $\EE\left[\left(X_T^{(a)}\right)^{2+\varepsilon_1}\right]<\infty$. Applying Hölder's inequality for sums we get
\begin{align*}
    \EE\left[\left(\theta_u^{(a)}\right)^{2+2\varepsilon_2}\right]\leq\frac{2^{1+2\varepsilon_2}}{(1-\rho^2)^{1+\varepsilon_2}}\left[D^{1+\varepsilon_2}\EE\left[\left(\frac{1}{v_u}\right)^{1+\varepsilon_2}\right]+(a\rho)^{2+2\varepsilon_2}\EE\left[v_u^{1+\varepsilon_2}\right]\right]
\end{align*}
where $D=\sup_{t\in[0,T]}(\mu_t-r)^2$ is defined in Assumption \ref{as3}. By Proposition \ref{oldprop}(1) and \cite[Lemma 3.5]{arxiv2} we have that 
\begin{align*}
    \int_0^T\EE\left[\left(\theta_u^{(a)}\right)^{2+2\varepsilon_2}\right]du<\infty.
\end{align*}
Applying Hölder's inequality with $q_1=\frac{p_2}{p_2-1}>1$ and $p_2>1$ we get
\begin{align*}
    \int_0^T\EE[\lambda_{u}^{2p_1}]^{\frac{1}{p_1}}\EE\left[\left(\theta_u^{(a)}\right)^{2+2\varepsilon_2}\right]^{\frac{1}{p_2}}du\leq\left(\int_0^T\EE[\lambda_{u}^{2p_1}]^{\frac{q_1}{p_1}}du\right)^{\frac{1}{q_1}}\left(\int_0^T\EE\left[\left(\theta_u^{(a)}\right)^{2+2\varepsilon_2}\right]du\right)^{\frac{1}{p_2}}<\infty.
\end{align*}
By \eqref{Lcota} this implies that
\begin{align*}
    \int_0^T\EE\left[\lambda_{u}^2\left(\theta_u^{(a)}\right)^2\left(X_u^{(a)}\right)^2\right]du<\infty.
\end{align*}
Similarly, using again \cite[Section 3.1]{momentshawkes}, \cite[Lemma 3.5]{arxiv2} and Proposition \ref{oldprop}(1) we can show that the second term in \eqref{Lco} is finite. We conclude that the process $t\to\int_0^t\lambda_{u-}dX_u^{(a)}$ is a square integrable $(\mathcal{F},\PP)$-martingale.

(2) By \cite[Lemma A.8]{arxiv2} we know the process $t\to\int_0^tX_u^{(a)}d(N-\Lambda^N)_u$ is a square integrable $(\mathcal{F},\PP)$-martingale where $d\Lambda^N_t=\lambda_tdt$. Thus, 
\begin{align*}
  \EE\left[\int_s^tX_u^{(a)}dN_u\Big|\mathcal{F}_s\right]&=\EE\left[\int_s^tX_u^{(a)}d(N-\Lambda^N)_u\Big|\mathcal{F}_s\right]+\EE\left[\int_s^tX_u^{(a)}d\Lambda^N_u\Big|\mathcal{F}_s\right] \\
  & = \int_s^t\EE\left[X_u^{(a)}\lambda_u|\mathcal{F}_s\right]du.
\end{align*}
Recall that $d\lambda_t=-\beta(\lambda_t-\lambda_0)dt+\alpha dN_t$. Then,
\begin{align*}
    \EE\left[\int_s^tX_u^{(a)}d\lambda_u\Big|\mathcal{F}_s\right]&=-\beta\int_s^t\EE[X_u^{(a)}\lambda_u|\mathcal{F}_s]du+\beta\lambda_0(t-s)X_s^{(a)}+\alpha\EE\left[\int_s^tX_u^{(a)}dN_u\Big|\mathcal{F}_s\right] \\
    & =-(\beta-\alpha)\int_s^t\EE[X_u^{(a)}\lambda_u|\mathcal{F}_s]du+\beta\lambda_0(t-s)X_s^{(a)}.
\end{align*}
Applying Itô formula, using that $\lambda$ is of finite variation and $X^{(a)}$ is continuous we have
\begin{align*}
    \lambda_tX^{(a)}_t=\lambda_sX_s^{(a)}+\int_s^t\lambda_{u-}dX_u^{(a)}+\int_s^tX_u^{(a)}d\lambda_u.
\end{align*}
Using that $t\to\int_0^t\lambda_{u-}dX_u^{(a)}$ is a square integrable $(\mathcal{F},\PP)$-martingale we obtain
\begin{align*}
    \EE[\lambda_tX_t^{(a)}|\mathcal{F}_s] & =  \lambda_sX_s^{(a)}+\EE\left[\int_s^t\lambda_{u-}dX_u^{(a)}\Big|\mathcal{F}_s\right]+\EE\left[\int_s^tX_u^{(a)}d\lambda_u\Big|\mathcal{F}_s\right] \\
    & = \lambda_sX_s^{(a)} -(\beta-\alpha)\int_s^t\EE[X_u^{(a)}\lambda_u|\mathcal{F}_s]du+\beta\lambda_0(t-s)X_s^{(a)}.
\end{align*}
So, writing $h_s^{(a)}(t):=\EE[\lambda_tX_t^{(a)}|\mathcal{F}_s]$ we have $h_s^{(a)}(t)=\lambda_sX_s^{(a)} -(\beta-\alpha)\int_s^th_s^{(a)}(u)du+\beta\lambda_0(t-s)X_s^{(a)}$. Taking the derivative
\begin{align*}
    h_s^{(a)}(t)^\prime=-(\beta-\alpha)h_s^{(a)}(t)+\beta\lambda_0X_s^{(a)}.
\end{align*}
Solving the differential equation with the initial condition $h_s(s)=\lambda_sX_s^{(a)}$ we obtain
\begin{align*}
    h_s^{(a)}(t)=\EE\left[X_t^{(a)}\lambda_t|\mathcal{F}_s\right]=X_s^{(a)}\left(\lambda_s-\frac{\beta\lambda_0}{\beta-\alpha}\right)e^{-(\beta-\alpha)(t-s)}+\frac{\beta\lambda_0X_s^{(a)}}{\beta-\alpha}.
\end{align*}
\end{proof}

\begin{prop}\label{AFOR}
Let $\QQ(a)\in\mathcal{E}_m(Q_1,2+\varepsilon_1)$ and define the forward variance by $\xi_{s}^{(a)}(t):=\EE^{\QQ(a)}[v_t|\mathcal{F}_s]$, $0\leq s\leq t\leq T$. Then,
\begin{enumerate}
    \item The explicit expression is given by
\begin{align*}
    \xi_{s}^{(a)}(t)= & \ D_1^{(a)}(t-s)v_s + D_2^{(a)}(t-s)\lambda_s+D_3^{(a)}(t-s),
\end{align*}
where
\begin{gather*}
D_1^{(a)}(h):=e^{-\kappa^{(a)}h}, \hspace{1cm} D_2^{(a)}(h)=C_1^{(a)}\left(e^{-(\beta-\alpha)h}-e^{-\kappa^{(a)}h}\right), \notag\\ 
D_3^{(a)}(h):=\left(\frac{C_2^{(a)}}{\kappa^{(a)}}-\Bar{v}^{(a)}\right)e^{-\kappa^{(a)}h} -\frac{C_2^{(a)}}{\beta-\alpha}e^{-(\beta-\alpha)h}+C_3^{(a)}, \notag\\
C_1^{(a)}:=\frac{\eta\EE[J_1]}{\kappa^{(a)}-(\beta-\alpha)}, \hspace{1cm} C_2^{(a)}:=\frac{\eta\EE[J_1]\beta\lambda_0}{\kappa^{(a)}-(\beta-\alpha)}, \hspace{1cm} C_3^{(a)}:=\frac{\eta\EE[J_1]\beta\lambda_0}{\kappa^{(a)}(\beta-\alpha)}+\Bar{v}^{(a)}. 
\end{gather*}

\item The dynamics of the process  $t\to\xi_{s}^{(a)}(t)$ is given by 
\begin{align}\label{Aeqdem3}
    d\xi_{s}^{(a)}(t) = \left[-\kappa^{(a)}\left(\xi_{s}^{(a)}(t)-\Bar{v}^{(a)}\right)+\eta\EE[J_1]\left(\left(\lambda_s-\frac{\beta\lambda_0}{\beta-\alpha}\right)e^{-(\beta-\alpha)(t-s)}+\frac{\beta\lambda_0}{\beta-\alpha}\right)\right]dt.
\end{align}
\item The dynamics of the process $s\to\xi_s(t)$ is given by
       \begin{align*}
        d\xi_s^{(a)}(t)= & \ \sigma D_1^{(a)}(t-s)\sqrt{v_s}dW_s^{\QQ(a)}+\eta D_1^{(a)}(t-s)\left(dL_s-\EE[J_1]\lambda_sds\right) \\ &+\alpha D_2^{(a)}(t-s)\left(dN_s-\lambda_sds\right).
    \end{align*}
\end{enumerate}
\end{prop}
\begin{proof}
   First we prove (2), then (1) and (3).

   (2) We can write
\begin{align*}
    v_t=v_s-\kappa^{(a)}\int_s^t(v_u-\Bar{v}^{(a)})du+\sigma\int_s^t\sqrt{v_u}dW_u^{\QQ(a)}+\eta (L_t-L_s),
\end{align*}
where $\kappa^{(a)}=\kappa+a\sigma$ and $\Bar{v}^{(a)}=\frac{k\Bar{v}}{k+a\sigma}$. 
By Theorem \ref{risk}. $\EE^{\QQ(a)}\left[\exp\left(\frac{\rho^2}{2}\int_0^Tv_udu\right)\right]<\infty$. In particular, $\EE^{\QQ(a)}\left[\int_0^Tv_sds\right]<\infty$ and the process $t\to\int_0^t\sqrt{v_u}dW_u^{\QQ(a)}$ is a $(\mathcal{F},\QQ(a))$-martingale. Therefore, 
\begin{align}\label{Afor}
    \xi_{s}^{(a)}(t)=v_s-\kappa^{(a)}\int_s^t\left(\xi_{s}^{(a)}(u)-\Bar{v}^{(a)}\right)du+\sigma\int_0^s\sqrt{v_u}dW_u^{\QQ(a)}+\eta(\EE^{\QQ(a)}[L_t|\mathcal{F}_s]-L_s).
\end{align}
By Theorem \ref{risk}, $X^{(a)}$ is a $(\mathcal{F},\PP)$-martingale and using \cite[Lemma A.8(2)]{arxiv2} 
\begin{align*}
    \EE^{\QQ(a)}[L_t|\mathcal{F}_s]&=\frac{\EE[L_tX^{(a)}_T|\mathcal{F}_s]}{\EE[X^{(a)}_T|\mathcal{F}_s]} \\
    & =\frac{\EE[\EE[L_tX^{(a)}_T|\mathcal{F}_t]|\mathcal{F}_s]}{X^{(a)}_s} \\
    & =\frac{\EE[L_tX^{(a)}_t|\mathcal{F}_s]}{X^{(a)}_s} \\
    & = L_s+\EE[J_1]\int_s^t\frac{\EE[\lambda_uX_u^{(a)}|\mathcal{F}_s]}{X_s^{(a)}}du.
\end{align*}
Inserting the last expression in \eqref{Afor} and using Lemma \ref{lemfor}(2) we get
\begin{align*}
    d\xi_{s}^{(a)}(t) & =  \left[-\kappa^{(a)}\left(\xi_{s}^{(a)}(t)-\Bar{v}^{(a)}\right)+\eta\EE[J_1]\frac{\EE[\lambda_tX_t^{(a)}|\mathcal{F}_s]}{X_s^{(a)}}\right]dt \\
    & = \left[-\kappa^{(a)}\left(\xi_{s}^{(a)}(t)-\Bar{v}^{(a)}\right)+\eta\EE[J_1]\left(\left(\lambda_s-\frac{\beta\lambda_0}{\beta-\alpha}\right)e^{-(\beta-\alpha)(t-s)}+\frac{\beta\lambda_0}{\beta-\alpha}\right)\right]dt.
\end{align*}
This proves \eqref{Aeqdem3}. 

(1) The solution of the previous ODE is given by
\begin{align*}
    \xi_{s}^{(a)}(t)=e^{-\kappa^{(a)}(t-s)}v_s+e^{-\kappa^{(a)} t}\int_s^te^{\kappa^{(a)} u}\left(\kappa^{(a)}\Bar{v}^{(a)}+f_s(u)\right)du
\end{align*}
where
\begin{align*}
    f_s(t)=\eta\EE[J_1]\left(\left(\lambda_s-\frac{\beta\lambda_0}{\beta-\alpha}\right)e^{-(\beta-\alpha)(t-s)}+\frac{\beta\lambda_0}{\beta-\alpha}\right).
\end{align*}
Now 
\begin{align*}
    \int_s^te^{\kappa^{(a)} u}\kappa^{(a)}\Bar{v}^{(a)}du=\Bar{v}^{(a)}\left(e^{\kappa^{(a)} t}-e^{\kappa^{(a)} s}\right)
\end{align*}
and
\begin{align*}
    &\int_s^te^{\kappa^{(a)} u}f_s(u)du =\eta\EE[J_1]\left(\lambda_s-\frac{\beta\lambda_0}{\beta-\alpha}\right)\int_s^te^{\kappa^{(a)} u-(\beta-\alpha)(u-s)}du+\frac{\eta\EE[J_1]\beta\lambda_0}{\beta-\alpha}\int_s^te^{\kappa^{(a)} u}du \\
    & = C_1^{(a)}\left(\lambda_s-\frac{\beta\lambda_0}{\beta-\alpha}\right)\left(e^{\kappa^{(a)} t-(\beta-\alpha)(t-s)}-e^{\kappa^{(a)} s}\right)+\left(C_3^{(a)}-\Bar{v}^{(a)}\right)\left(e^{\kappa^{(a)} t}-e^{\kappa^{(a)} s}\right).
\end{align*}
We conclude that
\begin{align}\label{truc}
    \xi_{s}^{(a)}(t)= & \ e^{-\kappa^{(a)}(t-s)}v_s +
    C_1^{(a)}\left(\lambda_s-\frac{\beta\lambda_0}{\beta-\alpha}\right)\left(e^{-(\beta-\alpha)(t-s)}-e^{-\kappa^{(a)}(t- s)}\right) \\
    &+C_3^{(a)}\left(1-e^{-\kappa^{(a)} (t-s)}\right) \notag\\
    = & \ D_1^{(a)}(t-s)v_s+D_2^{(a)}(t-s)\lambda_s+D_3^{(a)}(t-s)\notag, 
\end{align}
where one can check that
\begin{align*}
   -C_1^{(a)}&\frac{\beta\lambda_0}{\beta-\alpha}\left(e^{-(\beta-\alpha)h}-e^{-\kappa^{(a)}h}\right)+C_3^{(a)}\left(1-e^{-\kappa^{(a)}h}\right)= \\
    & = \left(\frac{C_1^{(a)}\beta\lambda_0}{\beta-\alpha}-C_3^{(a)}\right)e^{-\kappa^{(a)}h}-\frac{C_1^{(a)}\beta\lambda_0}{\beta-\alpha}e^{-(\beta-\alpha)h}+C_3^{(a)} \\
    & = \left(\frac{C_2^{(a)}}{\kappa^{(a)}}-\Bar{v}^{(a)}\right)e^{-\kappa^{(a)}h}-\frac{C_2^{(a)}}{\beta-\alpha}e^{-(\beta-\alpha)h}+C_3^{(a)}=D_3^{(a)}(h) .
\end{align*}

(3) For convenience we take the explicit expression of $\xi_s^{(a)}(t)$ given in \eqref{truc}, that is  
    \begin{align*}
        \xi_{s}^{(a)}(t)= & \ e^{-\kappa^{(a)}(t-s)}v_s +
    C_1^{(a)}\left(\lambda_s-\frac{\beta\lambda_0}{\beta-\alpha}\right)\left(e^{-(\beta-\alpha)(t-s)}-e^{-\kappa^{(a)}(t- s)}\right) \\
    &+C_3^{(a)}\left(1-e^{-\kappa^{(a)} (t-s)}\right).
    \end{align*}
    Then, 
    \begin{align}\label{Aeq1}
        d\xi_{s}^{(a)}(t) = & \ \kappa^{(a)}e^{-\kappa^{(a)}(t-s)}v_sds+e^{-\kappa^{(a)}(t-s)}dv_s+C_1^{(a)}\left(e^{-(\beta-\alpha)(t-s)}-e^{-\kappa^{(a)}(t-s)}\right)d\lambda_s \notag\\
        & +C_1^{(a)}\left(\lambda_s-\frac{\beta\lambda_0}{\beta-\alpha}\right)\left[\left(\beta-\alpha\right)e^{-(\beta-\alpha)(t-s)}-\kappa^{(a)}e^{-\kappa^{(a)}(t-s)}\right]ds \notag\\
        &-C_3^{(a)}\kappa^{(a)}e^{-\kappa^{(a)}(t-s)}ds \notag\\
        = & \ \kappa^{(a)}e^{-\kappa^{(a)}(t-s)}v_sds+e^{-\kappa^{(a)}(t-s)}\left[-\kappa^{(a)}\left(v_s-\Bar{v}^{(a)}\right)ds+\sigma\sqrt{v_s}dW_s^{\QQ(a)}+\eta dL_s\right] \notag\\
        & +C_1^{(a)}\left(e^{-(\beta-\alpha)(t-s)}-e^{-\kappa^{(a)}(t-s)}\right)\left[-\beta(\lambda_s-\lambda_0)ds+\alpha dN_s\right] \notag\\
        & +C_1^{(a)}\left(\lambda_s-\frac{\beta\lambda_0}{\beta-\alpha}\right)\left[\left(\beta-\alpha\right)e^{-(\beta-\alpha)(t-s)}-\kappa^{(a)}e^{-\kappa^{(a)}(t-s)}\right]ds \notag\\
        &-C_3^{(a)}\kappa^{(a)}e^{-\kappa^{(a)}(t-s)}ds.
    \end{align}
    Note that the $ds$ coefficients are
    \begin{align}\label{Auglyexpr}
        & \ \kappa^{(a)}e^{-\kappa^{(a)}(t-s)}v_s -e^{-\kappa^{(a)}(t-s)}\kappa^{(a)}\left(v_s-\Bar{v}^{(a)}\right)+C_1^{(a)}\left(e^{-(\beta-\alpha)(t-s)}-e^{-\kappa^{(a)}(t-s)}\right)\left[-\beta(\lambda_s-\lambda_0)\right] \notag\\
        & +C_1^{(a)}\left(\lambda_s-\frac{\beta\lambda_0}{\beta-\alpha}\right)\left[\left(\beta-\alpha\right)e^{-(\beta-\alpha)(t-s)}-\kappa^{(a)}e^{-\kappa^{(a)}(t-s)}\right]-C_3^{(a)}\kappa^{(a)}e^{-\kappa^{(a)}(t-s)}.
    \end{align}
    One can check that the expression in \eqref{Auglyexpr} is, indeed, equal to, 
    \begin{align}\label{Abetterexpr}
        -\alpha C_1^{(a)}\left(e^{-(\beta-\alpha)(t-s)}-e^{-\kappa^{(a)}(t-s)}\right)\lambda_s-\eta\EE[J_1]e^{-\kappa^{(a)}(t-s)}\lambda_s.
    \end{align}
    Then, the dynamics of $s\to\xi_s(t)$ in \eqref{Aeq1} can be written as
    \begin{align*}
        d\xi_s^{(a)}(t)= & \ \sigma e^{-\kappa^{(a)}(t-s)}\sqrt{v_s}dW_s^{\QQ(a)}++ \eta e^{-\kappa^{(a)}(t-s)}\left(dL_s-\EE[J_1]\lambda_sds\right)\\
        &+\alpha C_1^{(a)}\left(e^{-(\beta-\alpha)(t-s)}-e^{-\kappa^{(a)}(t-s)}\right)\left(dN_s-\lambda_sds\right) \\ 
        = & \ \sigma D_1^{(a)}(t-s)\sqrt{v_s}dW_s^{\QQ(a)}+\eta D_1^{(a)}(t-s)\left(dL_s-\EE[J_1]\lambda_sds\right) \\ &+\alpha D_2^{(a)}(t-s)\left(dN_s-\lambda_sds\right),
    \end{align*}
    which would finish the proof. We now check that the expressions in \eqref{Auglyexpr} and \eqref{Abetterexpr} are the same. Note that the coefficients of $e^{-(\beta-\alpha)(t-s)}$ in \eqref{Auglyexpr} are 
    \begin{align*}
        & \ C_1^{(a)}\left[-\beta(\lambda_s-\lambda_0)\right]+C_1^{(a)}\left(\lambda_s-\frac{\beta\lambda_0}{\beta-\alpha}\right)(\beta-\alpha)=-\alpha C_1^{(a)}\lambda_s,
    \end{align*}
    which is equal to the coefficients of $e^{-(\beta-\alpha)(t-s)}$ in \eqref{Abetterexpr}. The coefficients of $e^{-\kappa^{(a)}(t-s)}$ in \eqref{Auglyexpr} are 
    \begin{align*}
       & \  \kappa^{(a)}\Bar{v}^{(a)}+\frac{\eta\EE[J_1]}{\kappa^{(a)}-(\beta-\alpha)}\beta(\lambda_s-\lambda_0)-\frac{\eta\EE[J_1]}{\kappa^{(a)}-(\beta-\alpha)}\left(\lambda_s-\frac{\beta\lambda_0}{\beta-\alpha}\right)\kappa^{(a)}
       \\ &-\left(\frac{\eta\EE[J_1]\beta\lambda_0}{\kappa^{(a)}(\beta-\alpha)}+\Bar{v}^{(a)}\right)\kappa^{(a)} \\
        = & \ \frac{\eta\EE[J_1]}{\left[\kappa^{(a)}-(\beta-\alpha)\right](\beta-\alpha)}\left[\beta(\lambda_s-\lambda_0)(\beta-\alpha)-\lambda_s\kappa^{(a)}(\beta-\alpha)+\beta\lambda_0\kappa^{(a)}-\beta\lambda_0\left(\kappa^{(a)}-(\beta-\alpha)\right)\right] \\
        = & \ \frac{C_1^{(a)}}{\beta-\alpha}\left[\beta(\lambda_s-\lambda_0)(\beta-\alpha)-\lambda_s\kappa^{(a)}(\beta-\alpha)+\beta\lambda_0(\beta-\alpha)\right] \\
        = & \ C_1^{(a)}\left[\beta(\lambda_s-\lambda_0)-\lambda_s\kappa^{(a)}+\beta\lambda_0\right] \\
        = & \ C_1^{(a)}\left[\alpha\lambda_s-\alpha\lambda_s+\beta(\lambda_s-\lambda_0)-\lambda_s\kappa^{(a)}+\beta\lambda_0\right] \\
        = & \ \alpha C_1^{(a)}\lambda_s-\frac{\eta\EE[J_1]}{\kappa^{(a)}-(\beta-\alpha)}\left(\kappa^{(a)}-(\beta-\alpha)\right)\lambda_s \\
        = & \ \alpha C_1^{(a)}\lambda_s-\eta\EE[J_1]\lambda_s.
    \end{align*}
    This proves that the expressions in \eqref{Auglyexpr} and \eqref{Abetterexpr} are the same. 
\end{proof}

\begin{prop}\label{AlemV}
       Let $0\leq s\leq t\leq T$ and $\QQ(a)\in\mathcal{E}_m(Q_1,2+\varepsilon_1)$. Then, the explicit expression of $V_s^{(a)}(t)$ is given by
      \begin{align*}
    V_{s}^{(a)}(t)= & \ K_1^{(a)}(t-s)v_s + K_2^{(a)}(t-s)\lambda_s+K_3^{(a)}(t-s),
\end{align*}
where
\begin{gather*}
K_1^{(a)}(h):=A_{\kappa^{(a)}}(h), \hspace{1cm} K_2^{(a)}(h)=C_1^{(a)}\left[A_{\beta-\alpha}(h)-A_{\kappa^{(a)}}(h)\right], \notag\\ 
K_3^{(a)}(h):=\left[\frac{C_2^{(a)}}{\kappa^{(a)}}-\Bar{v}^{(a)}\right]A_{\kappa^{(a)}}(h) -\frac{C_2^{(a)}}{\beta-\alpha}A_{\beta-\alpha}(h)+C_3^{(a)}
\end{gather*}
where  $C_1^{(a)},C_2^{(a)}$ and $C_3^{(a)}$ are given in \eqref{eq:const}.
\end{prop}
\begin{proof}
Using the explicit expression of $\xi_{s}^{(a)}$ given in \eqref{expl2} and the Definition \ref{intdef} we obtain
\begin{align*}
    V_s^{(a)}(t) = & \ \frac{1}{t-s}\int_s^t\xi_{s}^{(a)}(u)du \\
     = & \ \frac{1}{t-s}\int_s^tD_1^{(a)}(u-s)duv_s+\frac{1}{t-s}\int_s^tD_2^{(a)}(u-s)du\lambda_s+\frac{1}{t-s}\int_s^tD_3^{(a)}(u-s)du \\
     = & \ A_{\kappa^{(a)}}(t-s)v_s +C_1^{(a)}\left[A_{\beta-\alpha}(t-s)-A_{\kappa^{(a)}}(t-s)\right]\lambda_s \\
     &+\left[\frac{C_2^{(a)}}{\kappa^{(a)}}-\Bar{v}^{(a)}\right]A_{\kappa^{(a)}}(t-s) -\frac{C_2^{(a)}}{\beta-\alpha}A_{\beta-\alpha}(t-s)+C_3^{(a)} \\ 
     = & \ K_1^{(a)}(t-s)v_s + K_2^{(a)}(t-s)\lambda_s+K_3^{(a)}(t-s).
\end{align*}
\end{proof}
\begin{lem}\label{ineqA}
    Let $k_1,k_2,h>0$ with $k_1\neq k_2$. Then, 
    \begin{align*}
        \frac{A_{k_2}(h)-A_{k_1}(h)}{k_1-k_2}>0.
    \end{align*}
\end{lem}
\begin{proof}
    We can assume without loss of generality that $k_1>k_2$. Now,
\begin{align*}
\frac{A_{k_2}(h)-A_{k_1}(h)}{k_1-k_2}=\frac{1}{h}\int_0^h\frac{e^{-k_2u}-e^{-k_1u}}{k_1-k_2}du>\frac{1}{h}\int_0^h\frac{e^{-k_1u}-e^{-k_1u}}{k_1-k_2}du=0.
\end{align*}
    This finishes the proof.
\end{proof}

\newpage
\printbibliography

\end{document}